\documentclass[11pt,english]{article}

\newcommand{\ARXIV}[1]{}

\usepackage{amsmath,amssymb,amsthm}
\usepackage{multicol}
\usepackage[table]{xcolor}
\usepackage[margin=1in]{geometry}
\usepackage{graphicx,color}
\usepackage{enumitem}
\usepackage{fullpage}
\usepackage[noblocks]{authblk}
\usepackage{tcolorbox}
\usepackage{babel}
\usepackage{wrapfig}
\usepackage{MnSymbol}
\usepackage{mdframed}
\usepackage{mathtools}
\usepackage{floatrow}
\usepackage{multirow}
\usepackage{booktabs}
\usepackage{palatino}

\usepackage[ruled,linesnumbered,vlined]{algorithm2e}
\usepackage{tablefootnote}
\usepackage{thm-restate}
\usepackage{bbding}
\usepackage{pifont}
\usepackage{nicefrac}
\DeclareMathOperator{\polylog}{polylog}

\newcommand{\R}{\mathbb{R}}
\newcommand{\OO}{\widetilde{O}}
\newcommand{\CCC}{{\cal C}}
\newcommand{\OPT}{\mbox{\rm OPT}}
\newcommand{\diam}{\mbox{\rm diam}}

\newcounter{magicrownumbers}

\graphicspath{{./figures/}}

\usepackage{tablefootnote}

\definecolor{Darkblue}{rgb}{0,0,.8}
\definecolor{Brown}{cmyk}{0,0.61,1.,0.60}
\definecolor{Purple}{cmyk}{0.45,0.86,0,0}
\definecolor{Darkgreen}{rgb}{0.133,0.500,0.133}

\definecolor{MyGreen}{rgb}{0.200,0.500,0.200}
\renewcommand{\emph}[1]{{\color{MyGreen}{\em #1}}}

\usepackage[colorlinks,linkcolor=Darkblue,filecolor=blue,citecolor=blue,urlcolor=Darkblue,pagebackref]{hyperref}
\usepackage[nameinlink]{cleveref}

\usepackage[colorinlistoftodos,prependcaption,textsize=tiny]{todonotes}

\newcommand{\namedref}[2]{\hyperref[#2]{#1~\ref*{#2}}}
\newcommand{\propref}[1]{\hyperref[#1]{property~(\ref*{#1})}}

\newcommand{\eps}{\varepsilon}
\newcommand{\CC}{\mathcal{C}}
\newcommand{\DS}{\mathcal{DS}}
\newcommand{\AAA}{\mathcal{A}}
\newcommand{\IGNORE}[1]{}
\newcommand{\Ex}{\mathbb{E}}
\newcommand{\Var}{\textrm{Var}}
\newcommand{\XX}{X}

\newtheorem*{theorem*}{Theorem}
\newtheorem{theorem}{Theorem}[section]
\newtheorem{lemma}[theorem]{Lemma}
\newtheorem{definition}[theorem]{Definition}

\newtheorem{corollary}[theorem]{Corollary}

\newtheorem*{question*}{Question}

\newtheorem{remark}[theorem]{Remark}

\newtheorem*{conjecture*}{Conjecture}

\newtheorem{fact}[theorem]{Fact}

\newcommand{\old}[1]{{}}

\hypersetup{
  colorlinks=true,
  linkcolor=red!70!black,
  linktoc=all,
  citecolor=blue,
    pdfpagemode=FullScreen,
}

\title{Fast Static and Dynamic Approximation Algorithms for Geometric Optimization Problems:\\ Piercing, Independent Set, Vertex Cover, and Matching}

\date{}

 \author{Sujoy Bhore\thanks{Department of Computer Science \& Engineering, Indian Institute of Technology Bombay, Mumbai, India.\\ Email: \href{sujoy@cse.iitb.ac.in}{sujoy@cse.iitb.ac.in}}
 \quad
 Timothy M. Chan\thanks{Department of Computer Science, University of Illinois at Urbana-Champaign. Email: \href{tmc@illinois.edu}{tmc@illinois.edu}.\\  
 Work supported in part by
NSF Grant CCF-2224271.}
}

\begin{document}
\maketitle

\begin{abstract}
We develop simple and general techniques to obtain faster (near-linear time) static approximation algorithms, as well as efficient dynamic data structures, for four fundamental geometric optimization problems: \emph{minimum piercing set} (\textsc{MPS}), \emph{maximum independent set} (\textsc{MIS}), \emph{minimum vertex cover} (\textsc{MVC}), and \emph{maximum-cardinality matching} (\textsc{MCM}).  
Highlights of our results include the following:

\begin{itemize}
    \item For $n$ axis-aligned boxes in any constant dimension $d$, we give an $O(\log \log n)$-approximation algorithm for \textsc{MPS} that runs in $O(n^{1+\delta})$ time for an arbitrarily small constant $\delta>0$. This significantly improves the previous $O(\log\log n)$-approximation algorithm by Agarwal, Har-Peled, Raychaudhury, and Sintos (SODA~2024), which ran in $O(n^{d/2}\mathop{\rm polylog} n)$ time.

    \item Furthermore, we show that our algorithm can be made fully dynamic with $O(n^{\delta})$ amortized update time.  Previously, Agarwal et al.~(SODA~2024) obtained dynamic results only in $\mathbb{R}^2$ and achieved only $O(\sqrt{n}\mathop{\rm polylog} n)$ amortized expected update time. 

    \item For $n$ axis-aligned rectangles in $\mathbb{R}^2$, we give an $O(1)$-approximation
    algorithm for \textsc{MIS} that runs in $O(n^{1+\delta})$ time. Our result significantly improves the running time of the celebrated algorithm by Mitchell (FOCS~2021) (which was about $O(n^{21})$), and answers one of his open questions.  Our algorithm can also be made fully dynamic with $O(n^{\delta})$ amortized update time.
    
    

    
    \item For $n$ (unweighted or weighted) fat objects in any constant dimension, we give a dynamic $O(1)$-approximation algorithm for \textsc{MIS} with $O(n^{\delta})$ amortized update time. 
    Previously, Bhore, N\"ollenburg, T\'oth, and Wulms (SoCG 2024) obtained efficient dynamic $O(1)$-approximation algorithms only for disks in $\mathbb{R}^2$ and  only in the unweighted setting.

    \item For $n$ axis-aligned rectangles in $\mathbb{R}^2$, we give 
    a dynamic $(\frac{3}{2}+\varepsilon)$-approximation algorithm for \textsc{MVC} with $O(\mathop{\rm polylog} n)$ amortized update time for any constant $\varepsilon>0$. Our static result improves the running time of Bar-Yehuda, Hermelin, and Rawitz (2011). For disks in $\mathbb{R}^2$ or hypercubes in any constant dimension, we give the first fully dynamic $(1+\varepsilon)$-approximation algorithm for \textsc{MVC} with $O(\mathop{\rm polylog}n)$  amortized update time. 

    \item For (monochromatic or bichromatic) disks in $\mathbb{R}^2$ or hypercubes in any constant dimension, we give the first fully dynamic $(1+\varepsilon)$-approximation algorithm for \textsc{MCM} with $O(\mathop{\rm polylog}n)$ amortized update time. 
\end{itemize}


\end{abstract}

\newpage
\setcounter{tocdepth}{2}
	\tableofcontents

\newpage	

\section{Introduction}


In this work, we study geometric versions of four fundamental optimization problems: \emph{minimum piercing set} (\textsc{MPS}), \emph{maximum independent set} (\textsc{MIS}), 
\emph{minimum vertex cover} (\textsc{MVC}), and 
\emph{maximum-cardinality matching} (\textsc{MCM}). 
The first three problems are NP-hard for most types of geometric objects, and polynomial-time approximation algorithms for these problems have been extensively studied in the computational geometry literature.  Recently, researchers have started exploring techniques to improve the running time of such approximation algorithms, e.g., for geometric set cover \cite{AgarwalP20, ChekuriHQ20,  ChanH20}, MPS \cite{AHRS23},
and MCM \cite{Har-PeledY22}.  Given the advent of large-scale datasets in today's world, algorithms with \emph{near-linear running time} are especially desirable.  Furthermore, many real-world problems require efficient processing of geometric data which undergo updates.  This has also prompted researchers to explore efficient
\emph{dynamic} approximation algorithms for these problems \cite{Henzinger0W20,  AgarwalCSXX22, ChanH21, ChanHSX22, AHRS23, BNTW24}.  (Note that the existence of efficient static algorithms with subquadratic running time is a prerequisite to the existence of efficient dynamic algorithms with sublinear update time.)  In this paper, we continue investigating these two research directions.

\paragraph{Minimum piercing set (\textsc{MPS}).}
Given a set $S$ of $n$ geometric objects in $\mathbb{R}^d$, a subset $P\subset \mathbb{R}^d$ is a \emph{piercing set} of $S$ 
if every object of $S$ contains at least one point of $P$. The \emph{minimum piercing set} (\textsc{MPS}) problem asks for a piercing set $P$ of the smallest size.
The problem has numerous applications, in facility location, wireless sensor networks, etc.~\cite{sharir1996rectilinear, huang2002approximation, ben2000obnoxious, nielsen2003maintenance}.
The problem may be viewed as a ``continuous''  version of geometric hitting set (where $P$ is constrained to be a subset of a given discrete point set rather than $\R^d$), and geometric hitting set in turn corresponds to geometric set cover in the dual range space.
In particular, by the standard greedy algorithm for set cover,
one can compute an $O(\log n)$-approximation to the minimum piercing set in polynomial time for any family of piercing set with constant description complexity (since it suffices to work with a discrete set of $O(n^d)$ candidate points).


For unit squares/hypercubes, unit disks/balls, or more generally, near-equal-sized \emph{fat} objects in $\R^d$ for any constant $d$, there are simple $O(n)$-time $O(1)$-approximation algorithms, whereas the well-known shifted grid strategy of Hochbaum and Maass~\cite{hochbaum1985approximation} gives a \textsf{PTAS}, computing a $(1+\eps)$-approximation in $n^{O(1/\eps^d)}$ time.
For fat objects of arbitrary sizes, a simple greedy algorithm yields an $O(1)$-approximation; the running time is naively quadratic, but  can be improved to $O(n\polylog n)$ in the case of fat objects in $\R^2$ by using range-searching data structures~\cite{efrat2000dynamic}.
Chan~\cite{Chan03} gave a separator-based
\textsf{PTAS} for arbitrary fat objects (in particular, arbitrary hypercubes and balls), running in $n^{O(1/\eps^d)}$ time
(see also~\cite{chan2005approximating} for another PTAS for the case of unit-height rectangles in $\R^2$).

For arbitrary boxes\footnote{Throughout this paper, all rectangles, boxes, squares, and hypercubes  are axis-aligned by default.} in $\mathbb{R}^d$ (which may not be fat), the current best polynomial-time approximation algorithm achieves approximation ratio\footnote{Throughout this paper, $\OPT$ denotes the optimal value to an optimization problem.}
$O(\log\log\OPT)$.  The approach is to solve the standard linear program (LP) relaxation for piercing/hitting set (either exactly, or approximately by a multiplicative weight update method~\cite{BronnimannG95}), and then round the LP solution via \emph{$\eps$-nets}---the $\log\log$ approximation ratio comes from combinatorial bounds  by Aronov, Ezra, and Sharir~\cite{aronov2009small}  on $\eps$-nets for $d\in\{2,3\}$ and by
Ezra~\cite{ezra2010note} on weak $\eps$-nets for $d\ge 4$.

All these methods have high polynomial running time, prompting the following questions:


\begin{center}
\emph{Do there exist near-linear-time sublogarithmic-approximation algorithms for \textsc{MPS} for various families of geometric objects? Do there exist similar dynamic algorithms with sublinear update time?}
\end{center}

In SODA 2024, Agarwal, Har-Peled, Raychaudhury, and Sintos~\cite{AHRS23} presented a faster randomized $O(\log \log \OPT)$-approximation algorithm with expected running time\footnote{
Throughout this paper, the $\OO$ notation hides polylogarithmic factors in $n$.
} $\OO(n^{d/2})$ for boxes in $\mathbb{R}^d$. Moreover, they showed the expected running time can be improved to near-linear but only when
$\OPT$ is smaller than $n^{1/(d-1)}$.
Furthermore, they studied the problem in the dynamic setting. For rectangles  in 
$\mathbb{R}^2$ they obtained $\OO(\sqrt{n})$ amortized expected time per insertion/deletion.  There has been no other prior work on dynamic MPS (ignoring the easy case of near-equal-sized fat objects, where a straightforward grid strategy yields $O(1)$-approximation in $O(1)$ update time). 

\paragraph{Our contributions to \textsc{MPS}.}
\begin{itemize}
\item For boxes in $\R^d$ for any constant $d$,
we present an $O(\log\log\OPT)$-approximation algorithm running in $O(n^{1+\delta})$ time (Theorem~\ref{thm:pierce:box}) for an arbitrarily small constant $\delta>0$.
The running time is a dramatic improvement over Agarwal et al.'s previous
$\OO(n^{d/2})$ bound~\cite{AHRS23} for all $d\ge 3$.
\item Furthermore, our $O(\log\log\OPT)$-approximation algorithm for boxes can be made dynamic with $O(n^\delta)$ amortized update time (Theorem~\ref{thm:pierce:box}) for any $d\ge 2$.  This is a significant  improvement over
Agarwal et al.'s previous $\OO(\sqrt{n})$ bound, which addressed only the $\R^2$ case.
\item For fat objects in $\R^d$ for any constant $d$ (assuming constant description complexity), 
we present an $O(1)$-approximation algorithm running in $O(n^{1+\delta})$ time (Theorem~\ref{thm:pierce:fat}).  Recall that a near-linear-time implementation of the $O(1)$-approximation greedy algorithm~\cite{efrat2000dynamic} was known only in $\R^2$; the exponent in the previous time bound converges to 2 as $d$ increases, due to the use of range searching.
\item Furthermore, our $O(1)$-approximation algorithm for fat objects can be made dynamic with $O(n^\delta)$ amortized update time (Theorem~\ref{thm:pierce:fat}) for any constant $d$.  No previous dynamic algorithms for fat objects were known even for the case of disks in $\R^2$
(Agarwal et al.~\cite{AHRS23} did consider the case of squares in $\R^2$ but obtained a weaker $\OO(n^{1/3})$ update time bound).

\end{itemize}














\paragraph{Maximum independent set (\textsc{MIS}).} 
Given a set $S$ of objects in $\R^d$,
the geometric \textsc{MIS} problem is to choose a maximum-cardinality subset $I\subseteq S$ of \emph{independent} (i.e., pairwise-disjoint) objects. 
The problem is among the most popular geometric optimization problems studied.
It is related to MPS: the
size of the MIS is always at most the size of the MPS; in fact, a standard LP for MIS is dually equivalent to the standard LP for MPS.


For near-equal-sized fat objects for any constant $d$,
Hochbaum and Maass's shifted grid method yields a \textsf{PTAS}~\cite{hochbaum1985approximation}. 
For fat objects of arbitrary sizes, a simple greedy algorithm yields an $O(1)$-approximation~\cite{efrat2000dynamic}, but several \textsf{PTAS}s with running time $n^{O(1/\eps^d)}$ or $n^{O(1/\eps^{d-1})}$ have been found, e.g.,
via shifted quadtrees, geometric separators, or local search
\cite{erlebach2005polynomial,Chan03,ChanH12}.

The case of arbitrary rectangles in $\R^2$ has especially garnered considerable attention.
A $(\log n)$-approximation $O(n\log n)$-time algorithm via straightforward binary divide-and-conquer has long been known~\cite{AgarwalKS98} (see also~\cite{KhannaMP98, Nielsen00}); by increasing the branching factor, the approximation ratio can be lowered to $\eps\log n$ with running time $n^{O(1/\eps)}$~\cite{Chan04, BermanDMR01}.  The first substantial progress was made by 
Chalermsook and Chuzhoy~\cite{ChalermsookC09} in SODA 2009, who obtained an $O(\log\log n)$-approximation polynomial-time algorithm for rectangles, by rounding the LP solution using an intricate analysis.
In a different direction, Adamaszek and Wiese~\cite{AdamaszekHW19} in SODA 2014 obtained a quasi-\textsf{PTAS}, i.e., a $(1+\eps)$-approximation algorithm running in $n^{\polylog n}$ time, by using separators and dynamic programming; their approach works more generally for polygons, in particular, arbitrary line segments (see also~\cite{FoxP11, CslovjecsekPW24} for other related results).  For rectangles,
Chuzhoy and Ene~\cite{ChuzhoyE16} improved this to a ``quasi-quasi-\textsf{PTAS}'', running in
$n^{\mathop{\rm polyloglog} n}$ time, with a more complicated algorithm.
In a remarkable breakthrough, Mitchell~\cite{mitchell2022approximating} obtained the first
polynomial-time $O(1)$-approximation algorithm for rectangles, by proving a variant of a combinatorial conjecture of Pach and Tardos~\cite{PachT00} and applying straightforward dynamic programming; the approximation ratio was 10 in the original paper, but was subsequently lowered to $2+\eps$ by G\'alvez, Khan, Mari, M\"omke, Pittu, and Wiese~\cite{galvez20223,GalvezIMPROVED}.  

The running time of Mitchell's algorithm is high:
the original paper stated a (loose) upper bound of $O(n^{21})$, and although the exponent is likely improvable somewhat with more effort, it is not clear how to get a more practical polynomial bound.
For example, even if the original version of Pach and Tardos's conjecture were proven, the dynamic program would still require at least $n^4$ table entries.  (The running time of G\'alvez et al.'s algorithm is even higher.)
%
At the end of his paper, Mitchell~\cite{mitchell2022approximating} specifically asked the following question:

\begin{center}
\emph{``Can the running time of a constant-factor approximation algorithm be improved significantly?''}
\end{center}

The dynamic version of geometric MIS was first studied by Henzinger, Neumann, and Wiese\\~\cite{Henzinger0W20}, who gave dynamic $O(1)$-approximation algorithms for 
squares/hypercubes with amortized polylogarithmic update time (see also \cite{bhore2020dynamic}).
For arbitrary boxes, Henzinger et al.'s algorithm requires approximation ratio $O(\log^{d-1}n)$.
Cardinal, Iacono, and Koumoutsos~\cite{CardinalIK21} designed dynamic algorithms for fat objects with sublinear worst-case update time, using range searching data structures, but the exponent in the update bound converges to 1 as $d$ increases (even for squares or disks in $\R^2$, their update bound is large, near $\OO(n^{3/4})$). Very recently, Bhore, N\"ollenburg, T\'oth, and Wulms~\cite{BNTW24} showed that for disks in $\mathbb{R}^2$, an $O(1)$-approximation can be maintained in expected polylogarithmic update time.  Their algorithm for disks used certain dynamic data structures for range/intersection searching (requiring generalizations of dynamic 3-dimensional convex hulls~\cite{Chan10,Chan20,KaplanMRSS20}); even if it could be extended to balls and fat objects in higher dimensions, the exponent would also converge to 1 for larger $d$ (since convex hulls have much larger combinatorial complexity as dimension exceeds 3).

\begin{center}
\emph{Do there exist dynamic $O(1)$-approximation algorithms with sublinear update time for rectangles in $\mathbb{R}^2$, or with (say) $O(n^{0.1})$ update time for fat objects in $\R^d$ for $d\ge 3$?
}
\end{center}

\paragraph{Our contributions to \textsc{MIS}.}
\begin{itemize}
\item For rectangles in $\R^2$, we present an $O(1)$-approximation algorithm running in $O(n^{1+\delta})$ time (Theorem~\ref{thm:indep:box}) for an arbitrarily small constant $\delta>0$.  The running time is a dramatic improvement over Mitchell's previous algorithm~\cite{mitchell2022approximating}.
\item Furthermore, our $O(1)$-approximation algorithms for rectangles can be made dynamic with $O(n^\delta)$ amortized update time (Theorem~\ref{thm:indep:box}).  In contrast, Henzinger et al.'s previous dynamic algorithm~\cite{Henzinger0W20} had $O(\log n)$ approximation ratio.
\item For fat objects in $\R^d$ for any constant $d$ (assuming constant description complexity), we present an $O(1)$-approximation algorithm running in $O(n^{1+\delta})$ time (Theorem~\ref{thm:indep:fat}).
\item Furthermore, our $O(1)$-approximation algorithm for fat objects can be made dynamic with $O(n^\delta)$ amortized update time  (Theorem~\ref{thm:indep:fat}).  The update time is a significant improvement over Cardinal et al.'s~\cite{CardinalIK21}.  Surprisingly, range searching structures are not needed.
\item Our approach extends to the \emph{weighted} case (computing  the maximum-weight independent set for a given set of weighted objects).  For weighted rectangles in $\R^2$, our static and dynamic algorithms have the same running time but with $O(\log\log n)$ approximation ratio---this matches the current best approximation ratio for polynomial-time algorithms due to
Chaler-msook and Walczak~\cite{ChalermsookW21} (a generalization of the LP-rounding approach by Chalermsook and Chuzhoy~\cite{ChalermsookC09}).  The previous dynamic algorithm for weighted rectangles was due to Henzinger et al.~\cite{Henzinger0W20} and had $O(\log n)$ approximation ratio.  For weighted fat objects in $\R^d$, we obtain the same result with $O(1)$ approximation ratio.  In contrast, Cardinal et al.'s and Bhore et al.'s previous dynamic algorithms~\cite{CardinalIK21,BNTW24} inherently do not work in the weighted setting.  We obtain the \emph{first} efficient data structures for MIS for weighted disks in $\R^2$ and other types of weighted fat objects.
\end{itemize}

\paragraph{Minimum vertex cover (MVC).} 
Given an undirected graph $G=(V, E)$, a subset of vertices $\XX\subseteq V$ is a \emph{vertex cover} if each edge has at least one endpoint in $\XX$. The \emph{minimum vertex cover} (MVC) problem asks for a vertex cover with the minimum cardinality.
In the geometric version of the problem, we are given a set $S$ of $n$ geometric objects in $\R^d$ and want the MVC of the intersection graph of $S$.
In the exact setting, MVC is equivalent to MIS, as the complement of an MIS
is an MVC\@. However, from the approximation perspective, a sharp dichotomy exists between the two problems; for instance,  MIS on general graphs cannot be approximated with ratio $n^{1-\eps}$ under standard hypotheses~\cite{Zuckerman07},
but there is a simple greedy 2-approximation algorithm for MVC: namely, just take a maximal matching\footnote{A \emph{matching} in a graph $G=(V,E)$ is a subset of edges $M\subseteq E$ such that no two edges in $M$ share a common endpoint. A matching $M$ is \emph{maximal} if for every edge $uv\in E\setminus M$,  either $u$ or $v$ is matched in $M$.} and output its vertices.
MVC has not received as much attention as MPS and MIS in the geometry literature, but is just as natural to study in geometric settings:
if we are given a set of geometric objects that are almost non-overlapping and we want to remove the fewest number of objects to eliminate all the intersections, this is precisely the MVC problem in the intersection graph.

Erlebach, Jansen, and Seidel~\cite{erlebach2005polynomial} gave the first \textsf{PTAS} for MVC for fat objects in $\R^d$ running in 
$n^{O(1/\eps^d)}$ time.
For rectangles in the plane,
Bar-Yehuda, Hermelin, and Rawitz~\cite{bar2011minimum}
obtained a $(\frac{3}{2}+\varepsilon)$-approximation algorithm.
Bar-Yehuda et al.'s work
made use of Nemhauser and Trotter's standard LP-based kernelization for MVC~\cite{nemhauser1975vertex}, which allows us to approximate the MVC by flipping to an MIS instance.
Following the same kernelization approach,
Har-Peled~\cite{SarielVC23} noted that 
the known quasi-\textsf{PTAS} for MIS for polygons~\cite{AdamaszekHW19} and quasi-quasi-\textsf{PTAS}
for MIS for rectangles~\cite{ChuzhoyE16} imply a quasi-\textsf{PTAS} for MVC for polygons and quasi-quasi-\textsf{PTAS} for MVC for rectangles.
Recently, Lokshtanov, Panloan, Saurabh, Xue, and Zehavi~\cite{LokshtanovP00Z24} gave the first polynomial-time algorithm for strings (which include arbitrary line segments and polygons) that achieves a constant approximation ratio strictly below 2, using the Nemhauser--Trotter kernel and a number of new ideas.


We are interested in improving the running time of  static algorithms as well as the dynamic version of geometric MVC\@.  
Dynamic MVC is a well-studied problem in the dynamic graph algorithms literature, under edge updates; e.g., see \cite{OnakR10, BhattacharyaHI18, BhattacharyaK19}.
%
%
However, these graph results are not directly applicable to the geometric setting, since the insertion/deletion of a single object may require as many as $\Omega(n)$ edge updates in the intersection graph in the worst case.
%
There has been no prior work on dynamic geometric MVC (ignoring the case of monochromatic, nearly-equal-sized fat objects, where it is not difficult to adapt the standard shifted grid strategy~\cite{hochbaum1985approximation} to maintain a $(1+\eps)$-approximation with $O(1)$ update time).


\begin{center}
\emph{Do there exist near-linear-time better-than-2-approximation algorithms for \textsc{MVC} for various families of geometric objects? Do there exist similar dynamic  algorithms with sublinear update time?}
\end{center}



\paragraph{Our contributions to \textsc{MVC}.} 
\begin{itemize}
\item For rectangles in $\R^2$, we  speed up Bar-Yehuda et al.'s $(\frac{3}{2}+\eps)$-approximation polynomial-time algorithm~\cite{bar2011minimum} to run in $O(n\polylog n)$ time, and at the same time obtain a dynamic $(\frac{3}{2}+\eps)$-approximation algorithm with $O(\polylog n)$ amortized update time (Corollary~\ref{cor:rect}).

\item For disks in $\R^2$ and fat boxes (e.g., hypercubes) in $\R^d$ for any constant $d$, we speed up the previous \textsf{PTAS}~\cite{erlebach2005polynomial} to run in $O(n\polylog n)$ time (ignoring dependence on $\eps$), and  at the same time obtain a dynamic $(1+\eps)$-approximation algorithm with $O(\polylog n)$ amortized update time (Corollaries~\ref{disk:MVC} and~\ref{cor:fat:MVC}).
The fact that the approximation ratio is $1+\eps$ is notable and interesting: 
none of the known dynamic algorithms has approximation ratio $1+\eps$ for the other geometric optimization problems such as MPS and MIS\footnote{
This is with good reason: for MIS, it is not possible to maintain a $(1+\eps)$-approximation in sublinear time for a sufficiently small $\eps$, even in the case of unit squares in $\R^2$, since the static problem has a lower bound of $n^{\Omega(1/\eps)}$ under ETH~\cite{marx2007optimality}.
}, except in 1-dimensional special cases for intervals~\cite{Henzinger0W20,bhore2020dynamic,compton2020new}.

\item Similar results hold for bichromatic disks and fat boxes, for MVC in their \emph{bipartite} intersection graph
(Corollaries \ref{MVC:bipartite:disk} and~\ref{MVC:bipartite:box}).
\end{itemize}

Our results on MVC can be generalized to other types of fat objects in $\R^d$ (e.g., balls in $\R^3$) but with a larger time bound, dependent on range searching (with exponent converging to 1 as $d$ increases).  However, this is unavoidable for MVC (in contrast to our results on MIS): 
any dynamic approximation algorithm for MVC needs to recognize whether the MVC size is zero, and so must
know whether the intersection graph is empty.  Dynamic range emptiness (maintaining a dynamic set of input points so that we can quickly decide whether a query object contains any input point) can be reduced to this problem, by inserting all the input points, and repeatedly inserting a query object and deleting it.

\paragraph{Maximum-cardinality matching (MCM).} 

Another closely related classical optimization problem on graphs is maximum-cardinality matching (MCM), where the objective is to find a \emph{matching} with the largest number of edges in a given undirected graph $G$. The problem is related to MVC: the size of an MVC is always at least the size of an MCM; for bipartite graphs, they are well-known to be equal.
In fact, the standard LP for MVC is the dual to the LP for MCM\@.
MCM is polynomial-time solvable:
the classical algorithm by Hopcroft and Karp~\cite{HopcroftK73} runs in $O(m\sqrt{n})$ time
for bipartite graphs  with $n$ vertices and $m$ edges,
and Vazirani's algorithm~\cite{Vazirani94} achieves
the same run time for general graphs. By recent breakthrough results~\cite{ChenKLPGS22}, MCM can be solved in $m^{1+o(1)}$ time for bipartite graphs. 
Earlier, Duan and Pettie~\cite{DuanP10} obtained $O(m)$-time $(1+\eps)$-approximation algorithms for general graphs.

MCM on geometric intersection graphs has received some attention.
Efrat, Itai, and Katz~\cite{efrat2001geometry} showed how to compute the (exact) MCM in bipartite unit disk graphs in $O(n^{3/2} \log n)$ time. Their algorithm works for other geometric objects; for example, it runs in $O(n^{3/2} \polylog n)$ time
for bipartite intersection graphs of arbitrary disks,
by using known dynamic data structures for disk intersection searching~\cite{KaplanKKKMRS22}.
Cabello, Cheng, Cheong, and Knauer~\cite{CabelloCCK24} improved the time bound to $O(n^{4/3+\eps})$ for unit disks, and also gave further exact subquadratic-time algorithms for other types of objects, using biclique covers in combination with the recent graph results~\cite{ChenKLPGS22}.  See also~\cite{bonnet2023maximum} for other special-case results.
Recently, Har-Peled and Yang~\cite{Har-PeledY22} presented near-linear time 
$(1+\varepsilon)$-approximation algorithms for MCM in (bipartite or non-bipartite) intersection graphs of arbitrary disks, among other things.

Dynamic MCM is a well-studied problem in the dynamic graph algorithms literature, under edge updates; e.g., see \cite{OnakR10, GuptaP13, BernsteinS15, BernsteinS16, PelegS16, Solomon16, BhattacharyaHI18,  Behnezhad23, BhattacharyaKSW23, AzarmehrBR24}.
%
However, these graph results are not directly applicable to the geometric setting, again since the insertion/deletion of a single object may cause many edge updates.
There has been no prior work on dynamic geometric MCM
(again ignoring the easier case of monochromatic, nearly-equal-sized fat objects).

\begin{center}
    \emph{Do there exist dynamic $(1+\varepsilon)$-approximation algorithm for MCM for various families of geometric objects with sublinear update time?}
\end{center}

\paragraph{Our contributions to \textsc{MCM}.} 
\begin{itemize}
\item For disks in $\R^2$ and fat boxes (e.g., hypercubes) in $\R^d$ for any constant $d$, we obtain a dynamic $(1+\eps)$-approximation algorithm with $O(\polylog n)$ amortized update time (Corollaries~\ref{cor:mcm:disk}, \ref{cor:mcm:box}, \ref{cor:mcm:disk2}, and~\ref{cor:mcm:box2}), in both the monochromatic and bichromatic (bipartite) cases.  This can be viewed as a dynamization of Har-Peled and Yang's static algorithms~\cite{Har-PeledY22}.
\end{itemize}

\paragraph{Our techniques for MPS and MIS.} 

A natural approach to get faster static algorithms or efficient dynamic algorithms is to take a known polynomial-time static algorithm and modify it.  This was indeed the approach taken originally by Chan and He~\cite{ChanH21} on dynamic geometric set cover, and more recently by Agarwal et al.~\cite{AHRS23} on dynamic MPS\@.  In these works, the previous algorithms were LP-based, and the bottleneck was in solving the LP, which was done using a multiplicative weight update (MWU) method.  The idea was to apply geometric data structuring (range searching) techniques to speed up each iteration of the MWU; if $\OPT$ is small, the number of iterations is small, but if $\OPT$ is large, we can switch to a different strategy (since we can tolerate a larger additive error).

We started our research by following the above strategy but end up discovering a different, better, and \emph{simpler} approach: namely, we directly reduce our problem to smaller instances, and just solve these subproblems by invoking the known polynomial-time static algorithm \emph{as a black box}!  This way, we do not even need to know how the previous static algorithm internally works.  (This is advantageous for MIS for rectangles, for example, since Mitchell's previous algorithm was not based on LP/MWU, and it is not clear how it can be sped up with data structures.)

More precisely: for MPS/MIS for rectangles, we first use a standard divide-and-conquer (similar to \cite{AgarwalKS98}) to reduce the problem to less complex instances where the rectangles are stabbable by a small number of horizontal lines and by a small number of vertical lines.  The divide-and-conquer causes the approximation ratio to increase, but by using a larger branching factor $n^\delta$, the increase is only by a constant factor.  For each such instance, we ``round'' the input rectangles to reduce the number of rectangles to $n^{O(\delta)}$ (more formally, we form $n^{O(\delta)}$ ``classes'' and map each rectangle to a ``representative'' element in its class); we can then solve each such subproblem in $n^{O(\delta)}$ time by the black box.  The key step is to show that input rounding increases the approximation ratio by only a constant factor; this combinatorial fact has simple proofs.  This idea of reducing the input size by rounding is somewhat reminiscent to the familiar notion of \emph{coresets}~\cite{AgarwalHV}, though we have not seen coresets used in the context of geometric MPS/MIS before.

For fat objects, we proceed similarly, except that we use a divide-and-conquer based on shifted quadtrees~\cite{Chan98}.

Surprisingly, this (embarrassingly) simple approach is sufficient to yield all our new results by MPS and MIS---for example, the description of our method for MPS for rectangles fits in under two pages, in contrast to the much lengthier solution by Agarwal et al.~\cite{AHRS23}.  A virtue of this approach is that dynamization now becomes almost trivial.

\paragraph{Our techniques for MVC and MCM.}

For MVC, we return to the approach of speeding up MWU using geometric data structures.
There have been previous works~\cite{ChekuriHQ20, ChanH21} on speeding up MWU for static and dynamic geometric set cover,
but thus far not for geometric MVC\@.  We show that geometric MVC is well-suited to this approach (in some ways, even more so than geometric set cover): interestingly, the right data structure for an efficient implementation of MWU turns out to involve a type of 
a \emph{dynamic generalized closest pair} problem, which Eppstein~\cite{Eppstein95} (see also \cite{Chan20}) has conveniently developed a technique for.  For MVC, the purpose of using MWU to solve the LP (approximately) is in computing a  Nemhauser--Trotter-style kernel~\cite{nemhauser1975vertex}, which allows us to reduce $n$ to $\le 2\,\OPT$ (roughly), after which we can flip to a MIS instance.  As another technical ingredient, we show that an approximate LP solution is sufficient for the kernelization.
To solve geometric MVC in the dynamic setting, we additionally use a standard trick: \emph{periodic rebuilding}.
As mentioned, when $\OPT$ is small, the number of iterations of the MWU is small, but when $\OPT$ is large,
we only need to rebuild the solution after a long stretch of $\eps\,\OPT$ updates.

For MCM, kernels for matching seem harder to compute.  Instead, 
in the bipartite case, we adapt an approach based on Hopcroft and Karp's classical matching algorithm~\cite{HopcroftK73}, which is known to
yield good approximation after a constant number of iterations.  We show how to implement the approximate version of Hopcroft and Karp's algorithm using range searching data structures.  Previously, Efrat, Itai, and Katz~\cite{efrat2001geometry} have already applied geometric data structures to speed up Hopcroft and Karp's exact algorithm, but their work was on the static case.
The dynamic setting is trickier and requires a more delicate approach.
In the general non-bipartite case, we need one more idea by Lotker, Patt-Shamir, and Pettie~\cite{LotkerPP15}  (also used in~\cite{Har-PeledY22}) to reduce the non-bipartite to the bipartite case in the approximate setting; we give a reinterpretation
of this technique in terms of \emph{color-coding}~\cite{AlonYZ95}, of independent interest.

Both our methods for geometric MVC and MCM are quite general, and work for any family of objects satisfying certain requirements (see Theorems \ref{thm:main} and~\ref{thm:main:mcm} for the general framework).






\section{Minimum Piercing Set (MPS)}\label{sec:mps}


In this section, we present our static and dynamic approximation algorithms for MPS for boxes and fat objects.

\subsection{Boxes}

To solve the MPS problem for boxes, we first consider a special case that  can be solved by ``rounding'' the input boxes---this simple idea will be the key:

\begin{lemma}\label{lem:pierce:box}
Let $d$ be a constant.
Let $\Gamma$ be a set of $O(b)$ axis-aligned hyperplanes in $\R^d$.
Let $S$ be a set of $n$ axis-aligned boxes in $\R^d$ with the property that each box in $S$ is stabbed by at least one hyperplane in $\Gamma$ orthogonal to the $k$-th axis for every $k\in\{1,\ldots,d\}$.
We can compute an $O(\log\log \OPT)$-approximation to the minimum piercing set for $S$ in $\OO(n+b^{O(1)})$ time.
Furthermore, we can support insertions and deletions in $S$ (assuming the property) in $\OO(b^{O(1)})$ time.
\end{lemma}
\begin{proof}
The hyperplanes in $\Gamma$ form a (non-uniform) grid with $O(b^d)$ grid cells.
Place two boxes of $S$ in the same  \emph{class} if they intersect the same subset of grid cells (see Figure~\ref{fig:class}(a) for a depiction of one class).
There are $O(b^{2d})$ classes (as each class can be specified by $2d$ hyperplanes in $\Gamma$).
Let $\hat{S}$ be a subset of $S$ where we keep one ``representative'' element from each class.
Then $|\hat{S}|=O(b^{2d})$.
We apply the known result by Aronov, Ezra, and Sharir~\cite{aronov2009small} for $d\in \{2,3\}$ or Ezra~\cite{ezra2010note} for $d\ge 4$
(see also~\cite{AHRS23}) to compute an $O(\log\log\OPT)$-approximation
to the minimum piercing set for the boxes in $\hat{S}$.  This takes time polynomial in $|\hat{S}|$, i.e., $b^{O(1)}$ time.
Let $P$ be the returned piercing set for~$\hat{S}$.
For each point $p\in P$, add the $2^d$ corners of the grid cell containing $p$ to a set $P'$.
Then $|P'|\le 2^d|P|\le O(\log\log \OPT)\cdot\OPT$.  We output $P'$.

To show correctness, it suffices to show that $P'$ is a piercing set for $S$.
This follows because if $\hat{s}$ is the representative element of $s$'s class, and
$\hat{s}$ is pierced by $p$, then $s$ intersects the grid cell containing $p$ and so $s$ must be
pierced by one of the corners of the grid cell (because of the stated property), as illustrated in Figure~\ref{fig:class}(b).

Insertions and deletions are straightforward, by just maintaining a linked list per class, and re-running Agarwal et al.~\cite{AHRS23}'s
algorithm on $\hat{S}$ from scratch each time.
\end{proof}

\begin{figure}
    \centering
    \includegraphics{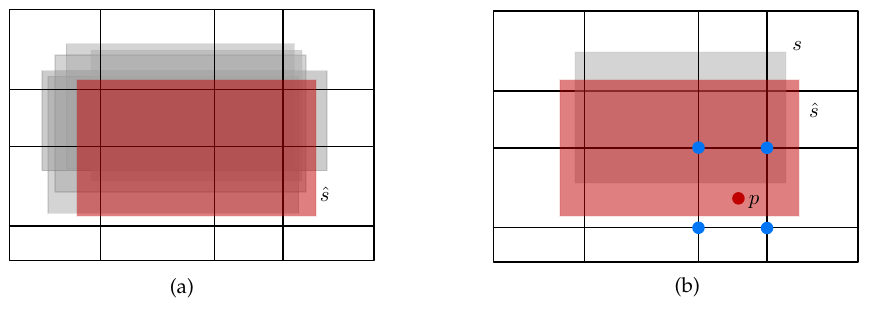}
    \caption{Proof of Lemma~\ref{lem:pierce:box}.}
    \label{fig:class}
\end{figure}

By combining the lemma with (a $b$-ary version of) a standard divide-and-conquer method~\cite{AgarwalKS98}, we obtain our main result for boxes:

\begin{theorem}\label{thm:pierce:box}
Let $d$ be a constant, and $\delta>0$ be a parameter.
Given a set $S$ of $n$ axis-aligned boxes in $\R^d$, 
we can compute an $O((1/\delta^d) \log\log \OPT)$-approximation to the minimum piercing set for $S$ in $O(n^{1+O(\delta)})$ time.
Furthermore, we can support insertions and deletions in $S$ in $O(n^{O(\delta)})$ amortized time.
%
\end{theorem}
\begin{proof}
In a \emph{type-$j$} problem ($j\in\{0,\ldots,d\}$), we assume that the given set $S$ is stabbable by $O(b)$ hyperplanes orthogonal to the $k$-th axis for every $k\in\{1,\ldots,j\}$.
The original problem is a type-0 problem.
A type-$d$ problem can be solved directly by Lemma~\ref{lem:pierce:box}.

To solve a type-$j$ problem with $j<d$, we build a $b$-ary tree\footnote{
This is basically a $b$-ary version of the \emph{interval tree}~\cite{preparata2012computational, de2000computational} applied to the projections of the boxes along the $(j+1)$-th axis.
} 
for $S$ as follows.  Pick $b-1$ hyperplanes orthogonal to the $(j+1)$-th axis to divide space into $b$ slabs, each containing
 $O(n/b)$ corner points.  Store the boxes in $S$ that are stabbed by the $b-1$ dividing hyperplanes at the root.  Recursively build subtrees for the subset of the $O(n/b)$ boxes in $S$ that are contained in each of the $b$ slabs.

For each of the $O(\log_b n)$ levels of the tree, we compute a piercing set for the subset of all boxes stored at that level;
we then return the union of these $O(\log_b n)$ piercing sets.  The approximation ratio consequently increases by an $O(\log_b n)$ factor.
To compute a piercing set at one level, since the boxes at different nodes lie in disjoint slabs, we can just
compute a piercing set for the boxes at each node separately and return the union.  
Computing a piercing set at each node reduces to a type-$(j+1)$ problem.

The approximation ratio for a type-$j$ problem satisfies the recurrence
$f_j \le O(\log_b n)\cdot f_{j+1}$, with $f_d=O(\log\log \OPT)$.
Thus, the overall approximation ratio is $f_0=O((\log_b n)^d \log\log \OPT)$.

To analyze the running time, observe that in a type-$j$ problem, each box is stored in one node of the tree and  is thus
assigned to one type-$(j+1)$ problem.
Since the running time for a type-$d$ problem is $\OO(b^{O(1)})$ per box by Lemma~\ref{lem:pierce:box}, the total running time is $\OO(b^{O(1)}n)$.

When we insert/delete a box in a type-$j$ problem, we insert/delete the box in one type-$(j+1)$ problem.
The update time for a type-$d$ problem is $\OO(b^{O(1)})$ by Lemma~\ref{lem:pierce:box}.
The update time for the original type-0 problem is thus $\OO(b^{O(1)})$.
One technical issue is tree balancing: for each a subtree with $n$ boxes, each child's slab should have $O(n/b)$ corner points.  We can use a standard weight-balancing scheme, rebuilding the subtree after it encounters $n/b$ updates.
The amortized cost for rebuilding is still $\OO(\frac{b^{O(1)}n}{n/b})=\OO(b^{O(1)})$ per level.

Finally, we set $b=n^\delta$ to get the bounds in the theorem.
\end{proof}

\begin{remark}\rm 
In Lemma~\ref{lem:pierce:box}, we know that $\OPT\le O(b^d)$ under the stated property, and an alternative proof of the lemma is to apply the MWU-based method by Agarwal et al.~\cite{AHRS23}, which is known to be efficient in the small $\OPT$ case.  However, our input rounding approach is more general and powerful (not limited to LP/MWU-based algorithms), and will be essential later, for our results on fat objects and on MIS.    
\end{remark}

\begin{remark}\rm \label{rmk:pierce}
Any improvement to the $O(\log\log\OPT)$ approximation ratio for polynomial-time algorithms for piercing boxes would automatically improve the approximation ratio for our static and dynamic algorithms.

For the static algorithm, the $1/\delta^d$ factor in the running time can be
lowered to $\log^d (1/\delta)$ by setting $b$ differently (as a function of the local input size $n$).  In particular, when setting $\delta=1/\log n$, we can obtain $O(n\polylog n)$ running time with approximation ratio $O(\log\log\OPT\cdot (\log\log n)^d)$ for boxes in $\R^d$.  In Appendix~\ref{app:samp}, we note another variant of the static algorithm with
$O(n\polylog n)$ running time, while keeping the approximation ratio at $O(\log\log\OPT)$; however, this variant uses Monte Carlo randomization, and works only in the setting when we want to output an approximation to the optimal value rather than an actual piercing set.  

For the dynamic algorithm, we have stated amortized bounds for simplicity; worst-case bounds seem plausible by standard deamortization techniques for weight-balanced trees~\cite{Overmars83}.
\end{remark}

\subsection{Fat objects}

We next turn to the case of fat objects.  We begin with some definitions and preliminary facts.
In what follows, the \emph{diameter} of an object $s$, denoted $\diam(s)$, refers to its $L_\infty$-diameter (i.e., the side length of its smallest enclosing axis-aligned hypercube).  We use the following definition of fatness~\cite{Chan03}, which is the most convenient here:

\begin{definition}\rm
A collection $\CCC$ of objects in $\R^d$ is \emph{$c$-fat} if the following property holds:
for every hypercube $B$, there exist $c$ points piercing all objects in $\CCC$ that intersect
$B$ and have diameter at least $\diam(B)$.
\end{definition}

\begin{definition}\rm
A \emph{quadtree box} is a hypercube of the form $[\frac{i_1}{2^\ell},\frac{i_1+1}{2^\ell})\times\cdots\times [\frac{i_d}{2^\ell},
\frac{i_d+1}{2^\ell})$
for some integers $i_1,\ldots,i_d,\ell\in \mathbb{Z}$.  (We also consider $\R^d$ to be a quadtree box.)
\end{definition}

The fact below follows by applying standard tree partitioning schemes, e.g., Frederickson~\cite{Frederickson85},\footnote{
Any constant-degree tree with $n$ nodes can be partitioned into $b$ connected pieces of $O(n/b)$ nodes each, such that each non-singleton
piece is adjacent to at most two other pieces~\cite{Frederickson85}.  When applied to the quadtree, each piece corresponds to a quadtree box
or the difference between two quadtree boxes.
}
to the (compressed) quadtree.

\begin{fact}\label{fact:part}
Let $d$ be a constant and $b$ be a parameter.
Let $P$ be a set of $n$ points in $\R^d$.  In $O(n)$ time,
we can partition of $\R^d$ into $b$ interior-disjoint cells, which each cell is either a quadtree box or the difference of two quadtree boxes, such that each cell contains at most $O(n/b)$ points of $P$.
\end{fact}

\begin{definition}\rm
An object $s$ is \emph{$c_0$-good} if it is contained in a quadtree box $B$ with $\diam(B)\le c_0\cdot \diam(s)$.
\end{definition}

\begin{fact}[Shifting Lemma~\cite{Chan98,Chan03}]\label{fact:shift}
Suppose $d$ is even.
Let $v_j=(\frac{j}{d+1},\ldots,\frac{j}{d+1})\in \R^d$.
For every object $s\subset [0,1)^d$, there exists $j\in\{0,\ldots,d\}$ such that $s+v_j$ is $O(d)$-good.
\end{fact}

We present our key lemma addressing a special case that can be solved by ``rounding'' the input objects:

\begin{lemma}\label{lem:pierce:fat}
Let $d,c,c_0$ be constants.
Let $\Gamma$ be a partition of $\R^d$ into $b$ disjoint cells, where each cell is either a quadtree box or the difference of two quadtree boxes.
Let $S$ be a set of $n$ $c_0$-good objects in $\R^d$ of constant description complexity from a $c$-fat collection $\CCC$,
with the property that each object in $S$ intersects the boundary of at least one cell of $\Gamma$.
We can compute an $O(1)$-approximation to the minimum piercing set for $S$ in $\OO(n+b^{O(1)})$ time.
Furthermore, we can support insertions and deletions in $S$ (assuming the property) in $\OO(b^{O(1)})$ time.
\end{lemma}
\begin{proof}
For each quadtree box $B$, define $\Lambda(B)$ to be a set of points piercing all $c_0$-good objects in $\CCC$ that intersect $\partial B$.   We can ensure that $|\Lambda(B)|=O(1)$ by fatness, since all $c_0$-good objects $s$ intersecting $\partial B$ have diameter more than $\diam(B)/c_0$, and we can cover $\partial B$ by
$O(1)$ hypercubes with diameter $\diam(B)/c_0$.

For each cell $\gamma$ which is the difference of an outer quadtree box $B^+$ with an inner quadtree box~$B^-$, define $\Lambda(\gamma)$ to be $\Lambda(B^+)\cup \Lambda(B^-)$.  Then $|\Lambda(\gamma)|=O(1)$.

It follows that $\OPT\le \sum_{\gamma\in\Gamma}|\Lambda(\gamma)|\le O(b)$.
In the static case, we could use the known greedy algorithm (e.g., see~\cite{Chan03, efrat2000dynamic})
to compute an $O(1)$-approximation to the minimum piercing set for the fat objects in $S$, which runs in time $\OO(n\cdot\OPT)=\OO(bn)$.  We propose a better approach which is dynamizable.

Place two objects of $S$ in the same  \emph{class} if they intersect the same subset of cells in $\Gamma$.
There are at most $b^{O(1)}$ classes: since the objects have constant description complexity, each object maps
to a point in a constant-dimensional space; the objects intersecting a cell map to a semialgebraic set in this space; 
a class corresponds to a cell in the arrangement of these $b$ semialgebraic sets; there are $b^{O(1)}$ cells in the arrangement.
We can determine the class of an object in $\OO(1)$ time by point location in the arrangement~\cite{AgarwalS00a}, after preprocessing in $b^{O(1)}$ time.
Let $\hat{S}$ be a subset of $S$ where we keep one ``representative'' element from each class.
Then $|\hat{S}|\le b^{O(1)}$.
We apply the known greedy algorithm to compute an $O(1)$-approximation to the minimum piercing set for 
the fat objects in $\hat{S}$.  This takes time polynomial in $|\hat{S}|$, i.e., $b^{O(1)}$ time.
Let $P$ be the returned piercing set for $\hat{S}$.
For each point $p\in P$, find the cell $\gamma\in\Gamma$ containing $p$ and add $\Lambda(\gamma)$ to a set $P'$.
Then $|P'|\le O(1)\cdot |P|\le O(1)\cdot\OPT$.  We output~$P'$.

To show correctness, it suffices to show that $P'$ is a piercing set for $S$.
This follows because if $\hat{s}$ is the representative element of $s$'s class, and
$\hat{s}$ is pierced by $p$, then $s$ intersects the cell $\gamma\in\Gamma$ containing $p$, and thus $s$ intersect
$\partial\gamma$ (because of the stated property), and so $s$ must be
pierced by one of the points in $\Lambda(\gamma)$.

Insertions and deletions are now straightforward, by just maintaining a linked list per class, and re-running Agarwal et al.'s
algorithm on $\hat{S}$ from scratch each time. 
\end{proof}

Combining with a quadtree-based divide-and-conquer, we obtain our main result for fat objects:

\begin{theorem}\label{thm:pierce:fat}
Let $d$ and $c$ be constants, and $\delta>0$ be a parameter.
Given a set $S$ of $n$ objects in $\R^d$ of constant description complexity from a $c$-fat collection $\CCC$, 
we can compute an $O(1/\delta)$-approximation to the minimum piercing set for $S$ in $O(n^{1+O(\delta)})$ time.
Furthermore, we can support insertions and deletions in $S$ in $O(n^{O(\delta)})$ amortized time.
%
\end{theorem}
\begin{proof}
We assume that all objects of $S$ are in $[0,1)^d$ and are $O(d)$-good.
This is without loss of generality by the shifting lemma (Fact~\ref{fact:shift}): for each of the $d+1$ shifts $v_j\ (j\in\{0,\ldots,d\})$, we can solve the problem for the good objects of $S+v_j$, and return the union of the piercing sets found (after shifting back by $-v_j$).
The approximation ratio increases by a factor of $d+1=O(1)$.

We build a $b$-ary tree\footnote{
This is basically a $b$-ary variant of Arya et al.'s \emph{balanced box decomposition (BBD) tree}~\cite{AryaMNSW98}.
} 
for $S$ as follows.  Arbitrarily pick one ``center'' point from each object of $S$ and apply
Fact~\ref{fact:part} to the $n$ center points of $S$, to obtain a partition $\Gamma$ into $b$ cells.
Store the objects that intersect the boundaries of these cells at the root.  Recursively build subtrees for the subset of the $O(n/b)$ boxes that are contained in each of the $b$ cells.

For each of the $O(\log_b n)$ levels of the tree, we compute a piercing set for the subset of all objects stored at that level;
we then return the union of these $O(\log_b n)$ piercing sets.  The approximation ratio consequently increases by an $O(\log_b n)$ factor.
To compute a piercing set at one level, since the objects at different nodes lie in disjoint cells, we can just
compute a piercing set for the boxes at each node separately and return the union.  
Computing a piercing set at each node reduces to the case handled by Lemma~\ref{lem:pierce:fat}.
The overall approximation ratio is thus $O(\log_b n)$.

To analyze the running time, observe that each object is stored in one node of the tree and  is thus
assigned to one subproblem handled by Lemma~\ref{lem:pierce:fat}.
The total running time is $\OO(b^{O(1)}n)$.

When we insert/delete an object, we insert/delete the object in one subproblem handled by Lemma~\ref{lem:pierce:fat}.
The update time is $\OO(b^{O(1)})$.
One technical issue is tree balancing: for each subtree with $n$ objects, each child's cell should have $O(n/b)$ center points.  We can use a standard weight-balancing scheme, rebuilding the subtree after encountering $n/b$ updates.
The amortized cost for rebuilding is still $\OO(\frac{b^{O(1)}n}{n/b})=\OO(b^{O(1)})$ per level.

Finally, we set $b=n^\delta$ to get the bounds in the theorem.
\end{proof}

Observations similar to Remark~\ref{rmk:pierce} hold here as well.

\section{Maximum Independent Set (MIS)}\label{sec:mis}

In this section, we present our static and dynamic approximation algorithms for MIS for (unweighted or weighted) rectangles and fat objects.  The approach is very similar to our algorithms for MPS in the previous section.  The main difference is in the justification that rounding the input objects increases the approximation ratio by  at most an $O(1)$ factor: the proofs are trickier, but still short.

\subsection{Rectangles}

\begin{lemma}\label{lem:indep:box}
Let $\Gamma$ be a set of $O(b)$ horizontal/vertical lines in $\R^2$.
Let $S$ be a set of $n$ axis-aligned rectangles in $\R^2$ with the property that each rectangle in $S$ is stabbed by at least one horizontal line 
and at least one vertical line in $\Gamma$.
We can compute an $O(1)$-approximation to the maximum independent set for $S$ in $\OO(n+b^{O(1)})$ time.
Furthermore, we can support insertions and deletions in $S$ (assuming the property) in $\OO(b^{O(1)})$ time.

If the rectangles in $S$ are weighted, we can do the same for an $O(\log\log b)$-approximation to the
maximum-weight independent set.
\end{lemma}
\begin{proof}
Define classes as in the proof of Lemma~\ref{lem:pierce:box}.
Let $\hat{S}$ be a subset of $S$ where we keep one ``representative'' element from each class---in the weighted case,
we keep the largest-weight element of the class.
Then $|\hat{S}|=O(b^4)$.
We apply Mitchell's result~\cite{mitchell2022approximating} (or its subsequent improvement~\cite{galvez20223,GalvezIMPROVED}) to compute an $O(1)$-approximation to the
maximum independent set for the rectangles in $\hat{S}$ in the unweighted case, or Chalermsook and Walczak's result~\cite{ChalermsookW21}
to compute an $O(\log\log|\hat{S}|)$-approximation
in the weighted case.  This takes time polynomial in $|\hat{S}|$, i.e., $b^{O(1)}$ time.
We output the returned independent set $\hat{I}$ for $\hat{S}$.

To analyze the approximation ratio,  let $I^*$ be the optimal independent set.
For each rectangle $s$, let $\hat{s}$ denote the representative element of $s$'s class.
We claim that $\{\hat{s}: s\in I^*\}$ contains an independent set of cardinality $\Omega(1)\cdot |I^*|$ in the unweighted case, or of weight $\Omega(1)$ times the weight of $I^*$ in the weighted case.  From the claim, it would follow that the overall approximation ratio is $O(1)$ in the unweighted case or $O(\log\log b)$ in the weighted case.

To prove the claim, we first show that $\{\hat{s}: s\in I^*\}$ has maximum depth\footnote{The \emph{depth} of a point is the number of objects containing the point.  For rectangles/boxes (but not necessarily other objects), the maximum depth (also called \emph{ply}) is equal to the maximum clique size in the intersection graph.}  $\Delta\le 4$.
To see this, observe that if a point $p$ lies inside $\hat{s}$, then $s$ intersects the
grid cell containing $p$, and so $s$ contains one of the 4 corners of this cell (because of the stated property), as illustrated in Figure~\ref{fig:class}(b), but there are at most 4
rectangles $s\in I^*$ satisfying this condition for a fixed $p$ (because of disjointness of $I^*$).
A classical result of Asplund and Gr\"unbaum~\cite{AG60} states that every arrangement of axis-aligned rectangles with maximum depth $\Delta$ is  
$O(\Delta^2)$-colorable.\footnote{
This has been improved to $O(\Delta\log\Delta)$ by Chalermsook and Walczak~\cite{ChalermsookW21}.
In our case, the rectangles are pseudo-disks, and the bound can be improved further to $O(\Delta)$.
But all this is not too important, since $\Delta=O(1)$ in our application.
}
Thus,  $\{\hat{s}: s\in I^*\}$ can be $O(1)$-colored, and the
largest-cardinality/weight color class, which is an independent set of $\hat{S}$,
must have an $\Omega(1)$ fraction of the cardinality/weight of $\{\hat{s}: s\in I^*\}$.


Insertions and deletions are straightforward, by just maintaining a linked list per class in the unweighted case, or a priority queue (to maintain the largest-weight element) per class in the weighted case, and re-running Mitchell's or Chalermsook and Walczak's
algorithm on $\hat{S}$ from scratch each time.
\end{proof}

\begin{proof}[Alternative Proof]
We describe an interesting, alternative proof of the above claim, which does not rely on the known coloring results.
Consider the grid formed by $\Gamma$.
For each rectangle $s$, let $\xi^-(s)$ and $\xi^+(s)$ be the grid columns containing the left and right edge of $s$ respectively, and let $\eta^-(s)$ and $\eta^+(s)$ be the grid rows containing the bottom and top edge of $s$ respectively.
Observe that there exists a subset $Z$ of the grid columns and rows, such that
$I^*_Z=\{s\in I^*: \xi^-(s)\in Z\ \wedge\ \xi^+(s)\not\in Z\ \wedge\ \eta^-(s)\in Z\ \wedge\ \eta^+(s)\not\in Z\}$ has at least $\frac{1}{16}$ of the cardinality/weight of $I^*$.
This can be proved in several ways\footnote{
This is similar to the well-known fact that in any undirected graph, the \emph{maximum cut} contains at least half of the edges (this has multiple proofs, including the simple probabilistic one).
}; for example, a standard, simple probabilistic argument is to pick $Z$ randomly and just note that the expected cardinality/weight of $I^*_Z$ is equal to $\frac{1}{16}$ times that of $I^*$.

To finish, observe that $\{\hat{s}: s\in I^*_Z\}$ is independent: if $s$ and $s'$ do not intersect but $\hat{s}$ and $\hat{s'}$ intersect, then $\xi^+(s)=\xi^-(s')$ or
$\xi^+(s')=\xi^-(s)$ or $\eta^+(s)=\eta^-(s')$ or $\eta^+(s')=\eta^-(s)$; but this can't happen when $s,s'\in I^*_Z$ by our definition of $I^*_Z$.
\end{proof}

\begin{theorem}\label{thm:indep:box}
Let $\delta>0$ be a parameter.
Given a set $S$ of $n$ axis-aligned rectangles in $\R^2$, 
we can compute an $O(1/\delta^2)$-approximation to the maximum independent set for $S$ in $O(n^{1+O(\delta)})$ time.
Furthermore, we can support insertions and deletions in $S$ in $O(n^{O(\delta)})$ amortized time.

If the rectangles in $S$ are weighted, we can do the same for an $O((1/\delta^2)\log\log n)$-approximation to the
maximum-weight independent set.
%
\end{theorem}
\begin{proof}
We proceed as in the proof of Theorem~\ref{thm:pierce:box}.  A type-$d$ problem is now
solved by Lemma~\ref{lem:indep:box} with $d=2$.

To solve a type-$j$ problem with $j<d$, we build the same $b$-ary tree as in the proof of Theorem~\ref{thm:pierce:box}.
For each of the $O(\log_b n)$ levels of the tree, we compute an independent set for the subset of all rectangles stored at that level;
we then return the largest-cardinality/weight of these $O(\log_b n)$ independent sets.  The approximation ratio consequently increases by an $O(\log_b n)$ factor.
To compute an independent set at a level, since the rectangles at different nodes lie in disjoint slabs, we can just
compute an independent set for the rectangles at each node separately and return the union.  
Computing an independent set at each node reduces to a type-$(j+1)$ problem.

The approximation ratio for a type-$j$ problem satisfies the recurrence
$f_j \le O(\log_b n)\cdot f_{j+1}$, with $f_d=O(1)$ in the unweighted case or $f_d=O(\log\log b)$ in the weighted case.
Thus, the overall approximation ratio is $f_0=O((\log_b n)^d)$ in the unweighted case or $f_0=O((\log_b n)^d \log\log b)$ with $d=2$.

The analysis of the running time and update time is as before.
%
\end{proof}

By standard binary divide-and-conquer~\cite{AgarwalKS98}, we can extend the result to higher-dimensional boxes, with the approximation ratio increased by one logarithmic factor per dimension:

\begin{corollary}\label{cor:indep:box}
Let $d$ be a constant and $\delta>0$ be a parameter.
Given a set $S$ of $n$ axis-aligned boxes in $\R^d$, 
we can compute an $O((1/\delta^2)\log^{d-2}n)$-approximation to the maximum independent set for $S$ in $O(n^{1+O(\delta)})$ time.
Furthermore, we can support insertions and deletions in $S$ in $O(n^{O(\delta)})$ amortized time.

If the boxes in $S$ are weighted, we can do the same for an $O((1/\delta^2)\log^{d-2}n\log\log n)$-approximation to the
maximum-weight independent set.
\end{corollary}

\begin{remark}\rm
Observations similar to Remark~\ref{rmk:pierce} hold here as well.  For example, in Appendix~\ref{app:rect}, we note a randomized variant of the static algorithm in $\R^2$ with $O(n\polylog n)$ running time and approximation ratio $O(1)$ (an absolute constant), when we only want an approximation to the optimal value but not an independent set.

Any improvement to the approximation ratio for polynomial-time algorithms for unweighted or weighted case would automatically imply analogous improvements  to the approximation ratio for our static and dynamic algorithms for any constant~$d$.  (Although the first proof of Lemma~\ref{lem:indep:box} relies on a coloring result that holds only in $\R^2$, the alternative proof of the lemma straightforwardly extends to higher dimensions.)
\end{remark}

\subsection{Fat objects}

\begin{lemma}\label{lem:indep:fat}
Let $d,c,c_0$ be constants.
Let $\Gamma$ be a partition of $\R^d$ into $b$ disjoint cells, where each cell is either a quadtree box or the difference of two quadtree boxes.
Let $S$ be a set of $n$ $c_0$-good weighted objects in $\R^d$ of constant description complexity from a $c$-fat collection $\CCC$, 
with the property that each object in $S$ intersects the boundary of at least one cell of $\Gamma$.
We can compute an $O(1)$-approximation to the maximum-weight independent set for $S$ in $\OO(n+b^{O(1)})$ time.
Furthermore, we can support insertions and deletions in $S$ (assuming the property) in $\OO(b^{O(1)})$ time.
\end{lemma}
\begin{proof}
Define $\Lambda(\cdot)$, classes, and $\hat{S}$ as in the proof of Lemma~\ref{lem:pierce:fat}.
We apply a known algorithm (e.g., see~\cite{Chan03}) to compute an $O(1)$-approximation to the maximum-weight independent set for 
the fat objects in $\hat{S}$.  This takes time polynomial in $|\hat{S}|$, i.e., $b^{O(1)}$ time.
We output the returned independent set $\hat{I}$ for $\hat{S}$.

To analyze the approximation ratio,  let $I^*$ be the optimal independent set.
For each object $s$, let $\hat{s}$ denote the representative element of $s$'s class.
We first show that $\{\hat{s}: s\in I^*\}$ has maximum depth $\Delta=O(1)$.
To see this, observe that if a point $p$ lies inside $\hat{s}$, then $s$ intersects
the cell $\gamma\in \Gamma$ containing $p$, and thus $s$ intersects $\partial\gamma$ (because of the stated property),
and so $s$ must contain at least one of the $O(1)$ points in $\Lambda(\gamma)$, but there are at most $O(1)$
objects $s\in I^*$ satisfying this condition (because of disjointness of $I^*$).
The intersection graph of any collection of $c$-fat objects with maximum depth $\Delta$ 
is $(c\Delta)$-degenerate\footnote{
Recall that a graph is \emph{$k$-degenerate} if every induced subgraph has a vertex of degree at most $k$.
To see why the intersection graph is $(c\Delta)$-degenerate, pick the object $s$ in the subgraph with the smallest diameter.  From the definition of $c$-fatness, the objects intersecting
$s$ can be pierced by $c$ points; so there can be at most $c\Delta$ objects intersecting $s$.  
} and is therefore $(c\Delta+1)$-colorable.
Thus,  $\{\hat{s}: s\in I^*\}$ can be $O(1)$-colored, and the
largest-weight color class, which is an independent set of $\hat{S}$,
must have $\Omega(1)$ fraction of the weight of $\{\hat{s}: s\in I^*\}$.
It follows that  the weight of $\hat{I}$ is at least $\Omega(1)$ times the weight of $I^*$.

Insertions and deletions are straightforward as before.
\end{proof}

\begin{theorem}\label{thm:indep:fat}
Let $d$ and $c$ be constants, and $\delta>0$ be a parameter.
Given a set $S$ of $n$ weighted objects in $\R^d$ of constant description complexity from a $c$-fat collection $\CCC$, 
we can compute an $O(1/\delta)$-approximation to the maximum-weight independent set for $S$ in $O(n^{1+O(\delta)})$ time.
Furthermore, we can support insertions and deletions in $S$ in $O(n^{O(\delta)})$ amortized time.
%
\end{theorem}
\begin{proof}
We proceed as in the proof of Theorem~\ref{thm:pierce:fat}.  As before, we assume that
all objects of $S$ are in $[0,1)^d$ and are $O(d)$-good.
This is without loss of generality by the shifting lemma (Fact~\ref{fact:shift}): for each of the $d+1$ shifts $v_j\ (j\in\{0,\ldots,d\})$, we can solve the problem for the good objects of $S+v_j$, and return the largest-weight of the independent sets found.
The approximation ratio increases by a factor of $d+1=O(1)$.

We build the same $b$-ary tree as in the proof of Theorem~\ref{thm:pierce:fat}.
For each of the $O(\log_b n)$ levels of the tree, we compute an independent set for the subset of all objects stored at that level;
we then return the largest-weight of these $O(\log_b n)$ independent sets.  The approximation ratio consequently increases by an $O(\log_b n)$ factor.
To compute an independent set at a level, since the objects at different nodes lie in disjoint cells, we can just
compute an independent set for the boxes at each node separately and return the union.  
Computing an independent set at each node reduces to the case handled by Lemma~\ref{lem:indep:fat}.
The overall approximation ratio is thus $O(\log_b n)$.

The analysis of the running time and update time is as before.
\end{proof}




\section{Minimum Vertex Cover (MVC)}
In this section, we study efficient static and dynamic algorithms for the MVC problem for intersection graphs of geometric objects.

\subsection{Approximating the LP via MWU}

\newcommand{\updateweight}{\textsc{update-weight}}
\newcommand{\findminedge}{\textsc{find-min-weight-edge}}

For a graph $G=(V,E)$, a \emph{fractional vertex cover} is a vector
$(x_v)_{v\in V}$ such that $x_u+x_v\le 1$ for all $uv\in E$ and $x_v\in [0,1]$ for all $v\in V$.  
Its \emph{size} is defined as $\sum_{v\in V}x_v$.  Finding a minimum-size fractional vertex cover
corresponds to solving an LP, namely, the standard LP relaxation of the MVC problem.

It is known that solving this LP is equivalent to computing the MVC in a related bipartite graph,
and thus can be done by known bipartite MCM algorithms---in fact,
in time almost linear in the number of edges by recent breakthrough results~\cite{ChenKLPGS22}. 
However, there are two issues that prevent us from applying such algorithms.
First of all, we are considering geometric intersection graphs, which may have $\Omega(n^2)$ number of edges; this issue could potentially be fixed by
using known techniques involving \emph{biclique covers} to sparsify the graph (maximum matching in a bipartite graph
then reduces to maximum flow in a sparser 3-layer graph~\cite{FederM95}). Second, for dynamic MVC, we will need efficient data
structures that can solve the LP still faster, in \emph{sublinear} time when $\OPT$ is small.

For our purposes, we only need to solve the LP approximately.
Our idea is to use a different well-known technique: \emph{multiplicative weight update} (MWU)\@. The key lemma is stated below.  The MWU algorithm and analysis here are not new (the description is short enough that we choose to include it to be self-contained),
and MWU algorithms have been used before for static and dynamic geometric set cover and other geometric optimization
problems (the application to vertex cover turns out to be a little simpler). However, our contribution is not in the proof of the lemma, but in the realization that MWU reduces the problem to designing a dynamic data structure (for finding min-weight edges subject to vertex-weight updates), which geometric intersection graphs happen to possess, as we will see.

\begin{lemma}\label{lem:mwu}
We are given a graph $G=(V,E)$.  Suppose there is a data structure $\DS$ for storing a vector $(w_v)_{v\in V}$
that can support the following two operations in $\tau$ time: (i) find an edge $uv\in E$ minimizing $w_u+w_v$, and (ii) update a number $w_v$.

Given a data structure $\DS$ for the vector that currently has $w_v=1$ for all $v\in V$,
we can compute a $(1+O(\delta))$-approximation to the minimum fractional vertex cover in
$\OO((1/\delta^2)\,\OPT\cdot \tau)$ time.  Here, $\OPT$ denotes the minimum vertex cover size.
\end{lemma}
\begin{proof}
Given a number $z$, the following algorithm attempts to find a fractional vertex cover of size at most~$z$ (below, $W$ denotes $\sum_{v\in V}w_v$):

\begin{quote}
\begin{tabbing}
for \= \kill
let $w_v=1$ for all $v\in V$, and $W=n$\\
while there exists $uv\in E$ with $w_u+w_v < W/z$ do \\
\> let $uv$ be such an edge\\
\> $W\leftarrow W + \delta(w_u+w_v)$,\ $w_u\leftarrow (1+\delta)w_u$,\ $w_v\leftarrow (1+\delta)w_v$
\end{tabbing}
\end{quote}

If and when the algorithm terminates, we have $w_u+w_v\ge W/z$ for all $uv\in E$.
Thus, defining $x_v := \min\{zw_v/W,\,1\}$, we have $x_u+x_v\ge 1$ for all $uv\in E$, and $\sum_{v\in V} x_v \le z$, i.e.,
$(x_v)_{v\in V}$ is a fractional vertex cover of size at most $z$.

We now bound the number of iterations $t$.
In each iteration, $W$ increases by at most a factor of $1+\delta/z$.  Thus,
at the end, 
\[ W\le (1+\delta/z)^t n.\]
Write each $w_v$ as $(1+\delta)^{c_v}$ for some integer~$c_v$.
Let $(x_v^*)_{v\in V}$ be an optimal fractional vertex cover of size $z^*$.
In each iteration, $\sum_{v\in V} c_vx_v^*$ increases by at least 1 (since we increment $c_u$ and $c_v$
for the chosen edge $uv$, and we know $x_u^*+x_v^*\ge 1$).
Thus, at the end, $\sum_{v\in V} c_vx_v^*\ge t$.
Since $\sum_{v\in V} x_v^* = z^*$,
it follows that $\max_{v\in V} c_v\ge t/z^*$.  Thus,
\[ W\ge (1+\delta)^{t/z^*}.\]
Therefore, $(1+\delta)^{t/z^*}\le (1+\delta/z)^tn\le e^{\delta t/z}n$, implying $(t/z^*)\ln (1+\delta) \le \delta t/z + \ln n$.
So, if $z\ge (1+\delta)z^*$, then $t\le \frac{z^*\ln n}{\ln (1+\delta) - \delta/(1+\delta)}= O((1/\delta^2)z^*\log n)$.

Note that only $O(t)=O((1/\delta^2)z^*\log n)$ of the numbers $w_v$ are not equal to 1, so 
the vector $(x_v)_{v\in V}$ can be encoded in $\OO(z^*)$ space.

We can try different $z$ values by binary or exponential search
till the algorithm terminates in $O((1/\delta^2)z^*\log n)$
iterations.  Each run can be implemented with $O((1/\delta^2)z^*\log n)$ operations in $\DS$.
After each run, we reset all the modified values $w_v$ back to 1 by $O((1/\delta^2)z^*\log n)$ update operations in $\DS$.
\end{proof}

\subsection{Kernel via Approximate LP}

Our approach for solving the MVC problem is to use a standard \emph{kernel} by
Nemhauser and Trotter~\cite{nemhauser1975vertex}, which allows us to reduce the problem to an instance where the number of vertices
is at most $2\,\OPT$.  Nemhauser and Trotter's construction is obtained from the LP solution: the kernel is 
simply the subset of all vertices $v\in V$ with $x_v=\frac12$.

In our scenario, we are only able to solve the LP approximately.  We observe that this is still enough to give a kernel
of approximately the same size.  We adapt the standard analysis of Nemhauser and Trotter, but
some extra ideas are needed.  The proof below will be self-contained.

\begin{lemma}\label{lem:kernel}
Let $c\ge 1$ and $0\le\delta<\gamma<\frac14$.
Given a $(1+O(\delta))$-approximation to the minimum fractional vertex cover in a graph $G=(V,E)$,
we can compute a subset $K\subseteq V$ of size at most $(2+O(\gamma))\,\OPT$, in $O(\OPT)$ time, such that 
a $c(1+O(\sqrt{\delta/\gamma}))$-approximation to the minimum vertex cover of $G$ can be obtained from
a $c$-approximation to the minimum vertex cover of $G[K]$ (the subgraph of $G$ induced by $K$).  Here, $\OPT$ denotes the minimum vertex cover size of $G$.
\end{lemma}
\begin{proof}
Let $(x_v)_{v\in V}$ be the given fractional vertex cover.  Let $\lambda < \gamma$ be a parameter to be set later.
Pick a value $\alpha\in [\frac12 - \gamma-\lambda,\frac12-\lambda]$ which is an integer multiple of $\lambda$,
minimizing $|\{v\in V: \alpha \le x_v < \alpha+\lambda\}|$.  This can be found in $O(|\{v\in V: x_v \ge \frac{1}{2}-\gamma-\lambda\}|)\le O(\OPT)$ time.\footnote{
This assumes an appropriate encoding of the vector $(x_v)_{v\in V}$, for example, the
encoding from the proof of Lemma~\ref{lem:mwu}.
}

Partition $V$ into 3 subsets: $L=\{v\in V: x_v < \alpha\}$, $H=\{v\in V: x_v > 1-\alpha\}$, and $K=\{v\in V: \alpha \le x_v \le 1-\alpha\}$.
Note that $|K|\le \frac{1}{\alpha}\sum_{v\in V}x_v\le (2+O(\gamma))(1+\delta)\OPT = (2+O(\gamma))\OPT$.

Let $\XX_K$ be a $c$-approximate minimum vertex cover of $G[K]$. 
We claim that $\XX := \XX_K\cup H$ is a $c(1+O(\sqrt{\delta/\gamma}))$-approximation to the minimum vertex cover of $G$.

First, $\XX_K\cup H$ is a vertex cover of $G$, since vertices in $L$ can only be adjacent
to vertices in $H$.

Let $\XX^*$ be a minimum vertex cover of $G$, with $|\XX^*|=\OPT$.
Since $K\cap \XX^*$ is a vertex cover of $G[K]$, we have
$|\XX_K|\le c|K\cap \XX^*|$, and so $|\XX|\le c(|K\cap \XX^*| + |H|) = c(|\XX^*| + |H\setminus \XX^*| - |L\cap \XX^*| )$.

We will upper-bound $|H\setminus \XX^*|-|L\cap \XX^*|$.
To this end, let $L'=\{v\in V: \alpha \le x_v < \alpha+\lambda\}$, and define the following modified vector $(x_v')_{v\in V}$:
\[ x'_v = \left\{\begin{array}{ll}
              x_v-\lambda & \mbox{if $v\in H\setminus \XX^*$}\\
              x_v+\lambda & \mbox{if $v\in (L\cup L')\cap \XX^*$}\\
              x_v  & \mbox{otherwise.}
\end{array}\right.
\]
Note that $(x'_v)_{v\in V}$ is still a fractional vertex cover, since for each edge $uv\in E$ with $u\in H\setminus \XX^*$ and $v\in (L\cup L')\cap \XX^*$,
we have $x'_u+x'_v = (x_u-\lambda)+(x_v+\lambda)\ge 1$; on the other hand, for each edge $uv\in E$ with $u\in H\setminus \XX^*$ and
$v\not\in (L\cup L')\cap \XX^*$,
we have $v\in \XX^*$ (since $\XX^*$ is a vertex cover) and so $v\not\in L\cup L'$,
implying that $x'_u+x'_v \ge (1-\alpha-\lambda)+(\alpha+\lambda)\ge 1$.
Now, $\sum_{v\in V} x_v - \sum_{v\in V} x'_v = \lambda (|H\setminus \XX^*| - |L\cap \XX^*| - |L'\cap \XX^*|)$.
On the other hand, $\sum_{v\in V} x_v - \sum_{v\in V} x'_v\le (1-\frac{1}{1+O(\delta)})\sum_{v\in V}x_v
\le O(\delta)|\XX^*|$.  It follows that $|H\setminus \XX^*| - |L\cap \XX^*| \le O(\frac{\delta}{\lambda})|\XX^*| + |L'\cap \XX^*|$.

By our choice of $\alpha$, we have $|L'|\le O(\frac{1}{\gamma/\lambda}) \cdot |\{v\in V: x_v \ge \frac{1}{2}-\gamma-\lambda\}|
\le O(\frac{\lambda}{\gamma})|\XX^*|$.
We conclude that $|\XX|\le c(|\XX^*| + |H\setminus \XX^*| - |L\cap \XX^*| ) \le c(1+O(\frac\delta\lambda + \frac\lambda\gamma)) |\XX^*|$.  Choose $\lambda=\sqrt{\gamma\delta}$.
\end{proof}

\subsection{Dynamic Geometric MVC via Kernels}

We now use kernels to reduce the dynamic vertex cover for geometric intersection graphs to a special case of
static vertex cover where the number of objects is approximately at most $2\,\OPT$.
We use a simple, standard idea for dynamization: be lazy, and periodically recompute the solution only after every $\eps\,\OPT$ updates
(since the optimal size changes by at most 1 per update).
This idea is commonly used in dynamic algorithms, e.g., in various previous works on dynamic geometric MIS~\cite{CardinalIK21},
dynamic MPS~\cite{AHRS23}, 
dynamic matching in graphs \cite{GuptaP13}, etc.

We observe that the data structure subproblem of dynamic min-weight intersecting pair, needed in Lemma~\ref{lem:mwu}, is reducible to 
dynamic intersection detection by known techniques.

\begin{lemma}\label{lem:Epp}
Let $\CC$ be a class of geometric objects, where there is a dynamic data structure $\DS_0$ for $n$ objects in $\CC$ that 
can detect whether there is an object 
intersecting a query object, and supports insertions and deletions of objects, with $O(\tau_0(n))$ query and update time.  

Then there is a dynamic data structure $\DS$ for $n$ weighted objects in $\CC$ that maintains
an intersecting pair of objects minimizing the sum of the weights, under insertions and deletions
of weighted objects, with $\OO(\tau_0(n))$ amortized time.
\end{lemma}
\begin{proof}
Eppstein~\cite{Eppstein95} gave a general technique to reduce the problem of dynamic closest pair to
the problem of dynamic nearest neighbor search, for arbitrary distance functions, while increasing
the time per operation by at most two logarithmic factors (with amortization).  Chan~\cite{Chan20} gave an alternative
method achieving a similar result.  The  $\DS$ problem can be viewed as a dynamic closest pair problem, where the distance between
objects $u$ and $v$ is $w_u+w_v$ if they intersect, and $\infty$ otherwise.
Thus, our problem reduces to the corresponding dynamic nearest neighbor search problem, namely, designing a data structure that can find a min-weight object intersecting
a query object, subject to insertions and deletions of objects (\emph{dynamic min-weight intersection searching}).

We can further reduce this to the $\DS_0$ problem (\emph{dynamic intersection detection}) by a standard \emph{multi-level} data structuring technique~\cite{AgarwalE99}, where
the primary data structure is a 1D search tree over the weights, 
and each node of the tree stores a secondary data structure for dynamic intersection detection.
Query and update time increase by one more logarithmic factor.
\end{proof}

We are now ready to state our general framework for solving MVC problems for geometric objects: the following theorem allows us to convert any efficient static approximation algorithm for the special case when $n$ is roughly less than $2\,\OPT$ to not only an efficient static approximation algorithm for the general case, but an efficient dynamic approximation algorithm at the same time, under the assumption that there
is an efficient dynamic data structure for intersection detection (as noted in the introduction, some form of range searching is unavoidable for MVC).

\begin{theorem}\label{thm:main}
Let $c\ge 1$, $\eps>0$, and $0\le\delta<\gamma<\frac14$.
Let $\CC$ be a class of geometric objects with the following oracles:
\begin{enumerate}
\item[(i)] a dynamic data structure $\DS_0$ for $n$ objects in $\CC$ that
can detect whether there is an object 
intersecting a query object, and supports insertions and deletions of objects, with $O(\tau_0(n))$ query and update time;
\item[(ii)] a static algorithm $\AAA$ for computing a $c$-approximation
of the minimum vertex cover of the intersection graph of $n$ objects in $\CC$ in $T(n)$ time, under the promise that 
$n\le (2+O(\gamma))\,\OPT$, where $\OPT$ is the optimal vertex cover size.
\end{enumerate}

Then there is a dynamic data structure for $n$ objects in $\CC$ that maintains a $c(1+O(\sqrt{\delta/\gamma}+\eps))$-approximation
of the minimum vertex cover of the intersection graph, under insertions and deletions of objects,
in $\OO((1/(\delta^2\eps)) \tau_0(n) + (1/\eps)T(n)/n)$ amortized time, assuming that $T(n)/n$ is monotonically increasing and $T(2n)=O(T(n))$.
\end{theorem}
\begin{proof}
Assume that $b\le\OPT < 2b$ for a given parameter $b$.
Divide into phases with $\eps b$ updates each.
At the beginning of each phase:
\begin{enumerate}
\item Compute a $(1+\delta)$-approximation to the minimum fractional vertex cover
by Lemma~\ref{lem:mwu} in $\OO((1/\delta^2)b\cdot \tau_0(n))$ amortized time using the data structure $\DS$ from Lemma~\ref{lem:Epp}.
(A weight change can be done by a deletion and an insertion.  It is important to note that we don't rebuild $\DS$
at the beginning of each phase; we continue using the same structure $\DS$ in the next phase, after resetting
the modified weights back to~1.)
\item Generate a kernel $K$ with size at most $(2+O(\gamma))\,\OPT$ by Lemma~\ref{lem:kernel} in $O(b)$ time.
\item Compute a $c$-approximation to the minimum vertex cover of the intersection graph of $K$ by the static algorithm $\AAA$ in
$O(T((2+O(\gamma))\,\OPT))=O(T(b))$ time, from which we obtain a $c(1+O(\sqrt{\delta/\gamma}))$-approximation
of the minimum vertex cover of the entire intersection graph.
\end{enumerate}
Since the above is done only at the beginning of each phase, the amortized cost per update is
$\OO(\frac{(1/\delta^2)b\cdot \tau_0(n) + T(b)}{\eps b})$.

During a phase, we handle an object insertion simply by inserting it to the current cover, and we handle an object deletion
simply by removing  it from the current cover.  This incurs an additive error at most $O(\eps b)=O(\eps\,\OPT)$.
We also perform the insertion/deletion of the object in $\DS$, with initial weight 1, in $\OO(\tau_0(n))$ amortized time.

How do we obtain a correct guess $b$?  We build the above data structure
for each $b$ that is a power of 2, and run the algorithm simultaneously for each $b$ (with appropriate cap on the run time based on $b$).
\end{proof}

\subsection{Specific Results}\label{sec:specific}

We now apply our framework to solve the dynamic geometric vertex cover problem
for various specific families of geometric objects.  By Theorem~\ref{thm:main}, it suffices to provide
(i)~a dynamic data structure $\DS_0$ for intersection detection queries,
and (ii)~a static algorithm $\AAA$ for solving the special case of the vertex cover problem
when $n\le (2+O(\gamma))\OPT$.

\paragraph{Disks in $\R^2$.}
Intersection detection queries for disks in $\R^2$ reduce to additively weighted Euclidean nearest neighbor search,
where the weights (different from the weights from MWU) are the radii of the disks.
Kaplan et al.~\cite{KaplanKKKMRS22} adapted Chan's data structure~\cite{Chan10,Chan20a} for dynamic Euclidean nearest neighbor search in $\R^2$ (reducible to dynamic convex hull in $\R^3$)
and obtained a data structure for
additively weighted Euclidean nearest neighbor search  with polylogarithmic amortized update time and query time. Thus,
$\DS_0$ can be implemented in $\tau_0(n)=O(\log^{O(1)}n)$ amortized time for disks in~$\R^2$.

In Appendix~\ref{app:disks}, we give a static $(1+O(\eps))$-approximation algorithm $\AAA$ for MVC for disks, under the promise that $n=O(\OPT)$, with running time $T(n)=\OO(2^{O(1/\eps^2)}n)$.
The algorithm is obtained by modifying a known {\sf PTAS} for MIS for disks~\cite{Chan03}.

Applying Theorem~\ref{thm:main} with $c=1+O(\eps)$, $\gamma=\Theta(1)$, and $\delta=\eps^2$, we obtain our main result for disks:

\begin{corollary}\label{disk:MVC}
There is a dynamic data structure for $n$ disks in $\R^2$
that maintains a $(1+O(\eps))$-approximation of the minimum vertex cover
of the intersection graph,  under insertions and deletions,
in $O(2^{O(1/\eps^2)}\log^{O(1)}n)$ amortized time.
(In particular, there is a static $(1+O(\eps))$-approximation algorithm running in 
$O(2^{O(1/\eps^2)}n\log^{O(1)}n)$ time.)
\end{corollary}

\paragraph{Rectangles in $\R^2$.}
For rectangles in $\R^2$, dynamic intersection detection requires $\tau_0(n)=O(\log^{O(1)}n)$ query and update time,
by standard orthogonal range searching techniques (range trees)~\cite{preparata2012computational,AgarwalE99}.

(Note: There is an alternative approach that bypasses Eppstein's technique and directly solves 
the $\DS$ data structure problem:  We dynamically maintain a \emph{biclique cover},
which can be done in polylogarithmic time for rectangles~\cite{Chan06}.  It is then easy to maintain the minimum weight
of each biclique, by maintaining the minimum weight of each of the two sides of the biclique with priority queues.)

In Appendix~\ref{app:rect}, we give 
a static $(\frac32 + O(\eps))$-approximation algorithm for MVC for rectangles, 
with running time $T(n)=\OO(2^{O(1/\eps^2)}n)$.  The algorithm is obtained by
modifying the method by Bar-Yehuda, Hermelin, and Rawitz~\cite{bar2011minimum}, and combining
with our efficient kernelization method.

Applying Theorem~\ref{thm:main} 
with $c=\frac32+O(\eps)$, $\gamma=\Theta(1)$, and $\delta=\eps^2$, 
we obtain our main result for rectangles:

\begin{corollary}\label{cor:rect}
There is a dynamic data structure for $n$ axis-aligned rectangles in $\R^2$
that maintains a $(\frac32+O(\eps))$-approximation of the minimum vertex cover
of the intersection graph, under insertions and deletions,
in $O(2^{O(1/\eps^2)}\log^{O(1)}n)$ amortized time.
(In particular, there is a static $(\frac32+O(\eps))$-approximation algorithm running in 
$O(2^{O(1/\eps^2)}n\log^{O(1)}n)$ time.)
\end{corollary}

\paragraph{Fat boxes (e.g., hypercubes) in $\R^d$.}

For the case of fat axis-aligned boxes (e.g., hypercubes) in a constant dimension~$d$,
dynamic intersection detection can again be solved with 
$\tau_0(n)=O(\log^{O(1)}n)$ amortized query and update time by orthogonal range searching~\cite{preparata2012computational,AgarwalE99}.
We can design the static algorithm $\AAA$ in exactly the same way as in the case of disks (since the method in Appendix~\ref{app:disks}
holds for fat objects in $\R^d$), achieving $T(n)=\OO(2^{O(1/\eps^d)}n)$.

\begin{corollary}\label{cor:fat:MVC}
There is a dynamic data structure for $n$ fat axis-aligned boxes in $\R^d$ for any constant $d$
that maintains a $(1+O(\eps))$-approximation of the minimum vertex cover
of the intersection graph, under insertions and deletions,
in $O(2^{O(1/\eps^d)}\log^{O(1)}n)$ amortized time.
\end{corollary}

We can also obtain results for balls or other types of fat objects in $\R^d$, but because intersection detection data structures
have higher complexity, the update time would be sublinear rather than polylogarithmic.

\paragraph{Bipartite disks in $\R^2$.}

For the case of a bipartite intersection graph between two sets of disks in $\R^2$,
we have $\tau_0(n)=O(\log^{O(1)}n)$ as already noted.
In the bipartite case, the static algorithm $\AAA$ is trivial: we just return the smaller of the two parts
in the bipartition, which yields a vertex cover of size at most $n/2\le (1+O(\eps))\,\OPT$ under the
promise that $n\le (2+O(\eps))\,\OPT$.

Applying Theorem~\ref{thm:main} with $c=1+O(\eps)$, $\gamma=\eps$, and $\delta=\eps^3$, we obtain:

\begin{corollary}\label{MVC:bipartite:disk}
There is a dynamic data structure for two sets of $O(n)$ disks in $\R^2$
that maintains a $(1+O(\eps))$-approximation of the minimum vertex cover
of the bipartite intersection graph, under insertions and deletions,
in $O((1/\eps^7)\log^{O(1)}n)$ amortized time.  
\end{corollary}

\paragraph{Bipartite boxes in $\R^d$.}

For the case of a bipartite intersection graph between two sets of (not necessarily fat) boxes in $\R^d$,
we have $\tau_0(n)=O(\log^{O(1)}n)$  by orthogonal range searching.
As already noted, in bipartite cases, the static algorithm $\AAA$ is trivial.

\begin{corollary}\label{MVC:bipartite:box}
There is a dynamic data structure for two sets of $O(n)$ axis-aligned boxes in $\R^d$ for any constant $d$
that maintains a $(1+O(\eps))$-approximation of the minimum vertex cover
of the bipartite intersection graph, under insertions and deletions,
in $O((1/\eps^7)\log^{O(1)}n)$ amortized time.  
\end{corollary}

We did not write out the number of logarithmic factors in our results, as we have not attempted to optimize them, but
it is upper-bounded by 3 plus the number of logarithmic factors 
in $\tau_0(n)$.

\section{Bipartite Maximum-Cardinality Matching (MCM)}

The approach in the previous section finds a $(1+O(\eps))$-approximation of the MVC
of bipartite geometric intersection graphs, and so it allows us to approximate the size of the MCM in such bipartite graphs.  However, it does not compute a matching.  In this section,
we give a different approach to maintain a $(1+O(\eps))$-approximation of the MCM in bipartite intersection graphs.
The first thought that comes to mind is to compute a kernel, as we have done for MVC, but for MCM,
known approaches seem to yield only a kernel of $O(\OPT^2)$ size (e.g., see~\cite{GuptaP13}).  Instead, we will
bypass kernels and construct an approximate MCM directly.

\subsection{Approximate Bipartite MCM via Modified Hopcroft--Karp}

It is well known that Hopcroft and Karp's $O(m\sqrt{n})$-time algorithm for exact MCM in bipartite graphs~\cite{HopcroftK73}
can be modified to give a $(1+\eps)$-approximation algorithm that runs in near-linear time, simply
by terminating early after $O(1/\eps)$ iterations (for example, see the introduction in~\cite{DuanP10}).
We describe a way to reimplement the algorithm in sublinear time
when $\OPT$ is small by using appropriate data structures, which correspond to dynamic intersection detection in
the case of geometric intersection graphs.  
Note that earlier work by Efrat, Itai, and Katz~\cite{efrat2001geometry} has already combined Hopcroft and Karp's algorithm with geometric data structures to obtain static exact algorithm~\cite{efrat2001geometry} for
maximum matching in bipartite geometric intersection graphs (see also Har-Peled and Yang's paper~\cite{Har-PeledY22} on static approximation algorithms).
However, to achieve bounds sensitive to $\OPT$, our algorithm will work differently (in particular, it will be DFS-based instead of BFS-based).

\newcommand{\maximalaugpaths}{\textsc{maximal-aug-paths}}
\newcommand{\visit}{\textsc{visit}}
\newcommand{\extend}{\textsc{extend}}

\begin{lemma}\label{lem:mcm}
We are given an unweighted bipartite graph $G=(V,E)$.  Suppose there is a data structure $\DS_0$ for storing a subset $\XX\subseteq V$ of vertices, initially with $\XX=\emptyset$, that supports the following two operations in $\tau_0$ time: given a vertex $u\in V$, find a neighbor of $u$ that is in $\XX$ (if exists), and insert/delete a vertex to/from~$\XX$.

Given a data structure $\DS_0$ that currently has $\XX=V$, and given a maximal matching $M_0$,
we can compute a $(1+O(\eps))$-approximation to the maximum-cardinality matching in  $\OO((1/\eps^2)\,\OPT\cdot \tau_0)$ time.  Here, $\OPT$ denotes the maximum matching size.
\end{lemma}
\begin{proof}
The algorithm proceeds iteratively.  We maintain a matching $M$.
At the beginning of the $\ell$-th iteration, we know that the current matching $M$ does not
have augmenting paths of length $\le 2\ell-1$.
We find a maximal collection $\Gamma$ of vertex-disjoint augmenting paths
of length $2\ell+1$.  We then augment $M$ along the paths in $\Gamma$.
As shown by Hopcroft and Karp~\cite{HopcroftK73}, the new matching $M$ will then not have augmenting paths of length $\le 2\ell+1$.

As shown by Hopcroft and Karp~\cite{HopcroftK73}, there are at least $\OPT-|M|$ vertex-disjoint augmenting paths, and so
$|M|\ge (\ell+1)(\OPT-|M|)$, i.e., $\OPT\le (1+\frac{1}{\ell+1})|M|$.  Thus, once $\ell$ reaches $\Theta(1/\eps)$,
we may terminate the algorithm.  Initially, we can set $M=M_0$ before the first iteration.

It suffices to describe how to find a maximal collection $\Gamma$ of vertex-disjoint augmenting paths
of length $2\ell+1$ in the $\ell$-th iteration, under the assumption that there are no shorter augmenting paths.
Hopcroft and Karp originally proposed a BFS approach, starting at the ``exposed'' vertices not covered by the current matching $M$.  Unfortunately, this approach does not work in our setting:
because we want the running time to be near $\OPT$, the searches need to start at vertices of $M$.
We end up adopting a DFS approach, but the vertices of $M$ need to be duplicated $\ell$ times in $\ell$ ``layers'' (this increases the running time by a factor of $\ell$, but fortunately, $\ell$ is small in our setting). 

Let $V_M$ be the $2|M|$ vertices of the current matching $M$.  As is well known, $\OPT\le 2|M|$.
A walk $v_0u_1v_1\cdots u_\ell v_\ell u_{\ell+1}$ in $G$ is an \emph{augmenting walk} of length $2\ell+1$
if $v_0\not\in V_M$, $v_0u_1\not\in M$, $u_1v_1\in M$, \ldots, $v_\ell u_{\ell+1}\not\in M$, and $u_{\ell+1}\not\in V_M$.
In such an augmenting walk of length $\ell$, we automatically have $v_0\neq u_{\ell+1}$ (because $G$ is bipartite)
and the walk must automatically be a simple path (because otherwise we could short-cut and obtain an augmenting path of length
$\le 2\ell-1$).
In the procedure $\extend(v_0u_1v_1\cdots u_i,\ell)$ below, the input is a walk $v_0u_1v_1\cdots u_i$ 
with $v_0u_1\not\in M$, $u_1v_1\in M$, \ldots, $v_{i-1} u_i\not\in M$, and the output is true if it is possible
to extend it to an augmenting walk $v_0u_1v_1\cdots u_\ell v_\ell u_{\ell+1}$  of length $2\ell+1$ that
is vertex-disjoint from the augmenting walks generated so far.

\begin{quote}
\begin{tabbing}
9.\ \ \ \= for \= for \= for \=\kill
$\maximalaugpaths(\ell)$:\\[.25ex]
1.\> let $\XX=V\setminus V_M$ and $\XX_1=\cdots=\XX_\ell=V_M$\\
2.\> for each $u_1\in \XX_1$ do\\
3.\>\> let $v_0$ be a neighbor of $u_1$ with $v_0\in \XX$\\
4.\>\> if $v_0$ does not exist then delete $u_1$ from $\XX_1$\\
5.\>\> else $\extend(v_0u_1,\ell)$\\[2ex]
$\extend(v_0u_1v_1\cdots u_i,\ell)$:\\[.25ex]
1.\> let $v_i$ be the partner of $u_i$ in $M$\\
2.\> if $i=\ell$ then\\
3.\>\> let $u_{\ell+1}$ be a neighbor of $v_\ell$ with $u_{\ell+1}\in \XX$\\
4.\>\> if $u_{\ell+1}$ does not exist then delete $u_\ell$ from $\XX_\ell$ and return false\\
5.\>\> output the augmenting path $v_0u_1v_1\cdots u_\ell v_\ell u_{\ell+1}$\\
6.\>\> delete $v_0$ and $u_{\ell+1}$ from $\XX$, and $u_1,\ldots,u_\ell$ from all of $\XX_1,\ldots,\XX_\ell$, and return true\\
7.\> for each neighbor $u_{i+1}$ of $v_i$ with $u_{i+1}\in \XX_{i+1}\setminus\{u_i\}$ do\\
8.\>\> if $\extend(v_0u_1v_1\cdots u_{i+1},\ell)=$ true then return true\\
9.\> delete $u_i$ from $\XX_i$ and return false
\end{tabbing}
\end{quote}

Note that if it is not possible to extend the walk $v_0u_1v_1\cdots u_i$ to an augmenting walk of length $2\ell+1$,
then it is not possible to extend any other walk $v_0'u_1'v_1'\cdots u_i'$ of the same length with $u_i'=u_i$.
This justifies why we may delete $u_i$ from $\XX_i$ in line~9 of $\extend$ (and similarly why we may delete $u_\ell$ from $\XX_\ell$
in line~4 of $\extend$, and why we may delete $u_1$ from $\XX_1$ in line~4 of $\maximalaugpaths$).

For the running time analysis, note that each vertex may be a candidate for $u_i$ in only one call to $\extend$ per $i$,
because we delete $u_i$ from $\XX_i$, either in line~9 if false is returned, or in line~6  if true is returned.
Thus, the number of calls to $\extend$ is at most $O(\ell|M|)$.

We use the given data structure $\DS_0$ to maintain $\XX$.  This allows us to do line~3 of $\maximalaugpaths$.
Line~1 of $\maximalaugpaths$ requires $O(|M|)$ initial deletions from $\XX$.  At the end, we 
reset $\XX$ to $V$ by performing $O(|M|)$ insertions of the deleted elements.

We also maintain $\XX_1,\ldots,\XX_\ell$ in $\ell$ new instances of the data structure $\DS_0$.
This allows us to do lines 3 and 7 of $\extend$.  Line~1 of $\maximalaugpaths$ requires $O(\ell|M|)$ initial
insertions to $\XX_1,\ldots,\XX_\ell$.
We conclude that $\maximalaugpaths$ takes $O(\ell |M|\cdot\tau_0)$ time.
Hence, the overall running time of all $\ell=\Theta(1/\eps)$ iterations is $O((1/\eps^2)|M|\cdot\tau_0)$.
\IGNORE{
\begin{quote}
\begin{tabbing}
9.\ \ \= for \= for \= for \=\kill
$\maximalaugpaths(\ell)$:\\[.25ex]
1.\> let $S=V\setminus V_M$ and $S_1=\cdots=S_\ell=V_M$\\
2.\> for each $u_1\in S_1$ do\\
3.\>\> let $v_0$ be a neighbor of $u_1$ with $v_0\in S$\\
4.\>\> if $v_0$ does not exist then delete $u_1$ from $S_1$\\
5.\>\> else if $\visit(u_1,1)=$ true then output $v_0u_1$, delete $v_0$ from $S$, and return true\\[2ex]
$\visit(u_i,i)$:\\[.25ex]
1.\> let $v_i$ be the partner of $u_i$ in $M$\\
2.\> if $i=\ell$ then\\
3.\>\> let $u_{\ell+1}$ be a neighbor of $v_\ell$ with $u_{\ell+1}\in S$\\
4.\>\> if $u_{\ell+1}$ does not exist then delete $u_\ell$ from $S_\ell$ and return false\\
5.\>\> output $u_\ell v_\ell$ and $v_\ell u_{\ell+1}$, delete $u_{\ell+1}$ from $S$ and $u_\ell$ from $S_1,\ldots,S_\ell$, and return true\\
6.\> for each neighbor $u_{i+1}$ of $v_i$ with $u_{i+1}\in S_{i+1}\setminus\{u_i\}$ do\\
7.\>\> if $\visit(u_{i+1},i+1)=$ true then output $u_i v_i$ and $v_i u_{i+1}$, delete $u_i$ from $S_1,\ldots,S_\ell$, and return true\\
8.\> delete $u_i$ from $S_i$ and return false
\end{tabbing}
\end{quote}
}
\end{proof}

\subsection{Dynamic Geometric Bipartite MCM} 

To solve dynamic bipartite MCM for geometric intersection graphs, we again use a standard idea for dynamization: be lazy, and periodically recompute the solution only after every $\eps\,\OPT$ updates
(since the optimal size changes by at most 1 per update).  The following theorem states our general framework:

\begin{theorem}\label{thm:main:mcm}
Let $\CC$ be a class of geometric objects, where there is
a dynamic data structure $\DS_0$ for $n$ objects in $\CC$ that
can find an object 
intersecting a query object (if exists), and supports insertions and deletions of objects, with $O(\tau_0(n))$ query and update time.

Then there is a dynamic data structure for two sets of $O(n)$ objects in $\CC$ that maintains a $(1+O(\eps))$-approximation
of the maximum-cardinality matching in the bipartite intersection graph, under insertions and deletions of objects,
in $\OO((1/\eps^3) \tau_0(n))$ amortized time.
\end{theorem}
\begin{proof}
First, observe that we can maintain a maximal matching $M_0$ with $O(\tau_0(n))$ update time:
We maintain the subset $S$ of all vertices not in $M_0$ in a data structure $\DS_0$.
When we insert a new object $u$, we match it with an object in $S$ intersecting $u$ (if exists) by querying $\DS_0$.
When we delete an object $u$, we delete $u$ and its partner $v$ in $M_0$, and reinsert $v$.

Assume that $b\le\OPT < 2b$ for a given parameter $b$.
Divide into phases with $\eps b$ updates each.
At the beginning of each phase, compute a $(1+O(\eps))$-approximation of the maximum-cardinality matching
by Lemma~\ref{lem:mcm} in $\OO((1/\eps^2)b\cdot \tau_0(n))$ time.
(It is important to note that we don't rebuild the data structure $\DS$ for $\XX=V$
at the beginning of each phase; we continue using the same structure $\DS$ in the next phase, after resetting
the modified $\XX$ back to $V$.)

Since the above is done only at the beginning of each phase, the amortized cost per update is
$\OO(\frac{(1/\eps^2)b\cdot \tau_0(n)}{\eps b})$.

During a phase, we handle an object insertion simply by doing nothing, and we handle an object deletion
simply by removing its incident edge (if exists) from the current matching.  This incurs an additive error at most $O(\eps b)=O(\eps\,\OPT)$.
We also perform the insertion/deletion of the object in $\DS_0$ for $\XX=V$, in $\OO(\tau_0(n))$ time.

How do we obtain a correct guess $b$?  We build the above data structure
for each $b$ that is a power of 2, and run the algorithm simultaneously for each $b$.
\end{proof}

\subsection{Specific Results}

Recall that for disks in $\R^2$ as well as boxes in $\R^d$,
we have $\tau_0(n)=O(\log^{O(1)}n)$.
(Note that most data structures for intersection detection can be modified to report a witness
object intersecting the query object if the answer is true.)
Thus, we immediately obtain:

\begin{corollary}\label{cor:mcm:disk}
There is a dynamic data structure for two sets of $O(n)$ disks in $\R^2$
that maintains a $(1+O(\eps))$-approximation of the maximum-cardinality matching
in the bipartite intersection graph, under insertions and deletions,
in $O((1/\eps^3)\log^{O(1)}n)$ amortized time.  
\end{corollary}

\begin{corollary}\label{cor:mcm:box}
There is a dynamic data structure for two sets of $O(n)$ axis-aligned boxes in $\R^d$ for any constant $d$
that maintains a $(1+O(\eps))$-approximation of the maximum-cardinality matching
in the bipartite intersection graph, under insertions and deletions,
in $O((1/\eps^3)\log^{O(1)}n)$ amortized time.  
\end{corollary}

\newcommand{\ZZ}{\mathcal{Z}}
\newcommand{\HH}{\mathcal{H}}

\section{General MCM}\label{sec:gen:match}

In this section, we adapt our results for bipartite MCM to
general MCM\@.  We use a simple idea
by Lotker, Patt-Shamir, and Pettie~\cite{LotkerPP15} to reduce the problem for general graphs (in the
approximate setting) to finding maximal collections of short augmenting paths in \emph{bipartite} graphs, which we already know how to solve.
(See also Har-Peled and Yang's paper~\cite{Har-PeledY22}, which applied Lotker et al.'s idea to obtain static approximation algorithms for  geometric intersection graphs.)  
The reduction has exponential dependence in the path length $\ell$ (which is fine since $\ell$ is small), and is originally randomized.
We reinterpret their idea in terms of \emph{color-coding}~\cite{AlonYZ95},
which allows for efficient derandomization, and also simplifies the analysis (bypassing Chernoff-bound calculations).  With this reinterpretation, it is easy to show that the idea carries over to the
dynamic setting.

We begin with a lemma, which is a consequence of the standard color-coding technique:

\begin{lemma}\label{lem:color:cod}
For any $n$ and $\ell$, there exists a collection $\ZZ^{(n,\ell)}$ of $O(2^{O(\ell)}\log n)$ subsets
of $[n] := \{1,\ldots,n\}$ such that
for any two disjoint sets $A,B\subseteq [n]$ of total size at most $\ell$, we have 
$A\subseteq Z$ and $B\subseteq [n]\setminus Z$ for some $Z\in\ZZ$.
Furthermore, $\ZZ^{(n,\ell)}$ can be constructed in $O(2^{O(\ell)}n\log n)$ time.
\end{lemma}
\begin{proof}
As shown by Alon, Yuster, and Zwick~\cite{AlonYZ95}, there exists a collection $\HH^{(n,\ell)}$
of $O(2^{O(\ell)}\log n)$ mappings $h:[n]\rightarrow [\ell]$,
such that for any set $\XX\subseteq [n]$ of size at most $\ell$,
the elements in $\{h(v): v\in \XX\}$ are all distinct.  Furthermore, $\HH^{(n,\ell)}$ can be constructed
in $O(2^{O(\ell)}n\log n)$ time.  (This is related to the notion of ``$\ell$-perfect hash family''.)

For each $h\in\HH^{(n,\ell)}$ and each subset $I\subseteq [\ell]$,
add the subset $Z_{h,I}=\{v\in[n]: h(v)\in I\}$ to $\ZZ^{(n,\ell)}$.
The number of subsets is $|\HH^{(n,\ell)}|\cdot 2^\ell \le 2^{O(\ell)}\log n$.
For any two disjoint sets $A,B\subseteq[n]$ of total size at most~$\ell$,
let $h\in\HH^{(n,\ell)}$ be such that the elements in $\{h(a): a\in A\}$ and $\{h(b): b\in B\}$
are all distinct,
and let $I=\{h(a): a\in A\}$; then $A\subseteq Z_{h,I}$
and $B\subseteq [n]\setminus Z_{h,I}$.
\end{proof}

We now present a non-bipartite analog of Lemma~\ref{lem:mcm}:

\begin{lemma}\label{lem:mcm:gen}
We are given an unweighted graph $G=(V,E)$, with $V=[n]$.  Let $\ZZ^{(n,1/\eps)}$ be as in Lemma~\ref{lem:color:cod}.
 Suppose there is a data structure $\DS_0^*$ for storing a subset $\XX\subseteq V$ of vertices, initially with $\XX=\emptyset$, that supports the following two operations in $\tau_0$ time: given a vertex $u\in V$ and $Z\in\ZZ^{(n,1/\eps)}$, find a neighbor of $u$ that is in $\XX\cap Z$
(if exists) and a neighbor of $u$ that is in $\XX\setminus Z$ (if exists); and insert/delete a vertex to/from~$\XX$.

Given a data structure $\DS_0$ that currently has $\XX=V$, and given a maximal matching $M_0$,
we can compute a $(1+O(\eps))$-approximation to the maximum-cardinality matching in  $\OO(2^{O(1/\eps)}\,\OPT\cdot \tau_0)$ time.  Here, $\OPT$ denotes the maximum matching size.
\end{lemma}
\begin{proof}
As in the proof of Lemma~\ref{lem:mcm}, we iteratively maintain a current matching $M$,
and it suffices to describe how to find a maximal collection $\Gamma$ of vertex-disjoint augmenting paths
of length $2\ell+1$ in the $\ell$-th iteration, under the assumption that there are no augmenting paths of length
$\le 2\ell-1$.
However, the presence of odd-length cycles complicates the computation of $\Gamma$.

Initialize $\Gamma=\emptyset$.  We loop through each $Z\in \ZZ^{(n,1/\eps)}$ one by one and do the following.  
Let $G_Z$ be the subgraph of $G$ with edges $\{uv\in E: u\in Z,\, v\not\in Z\}$.
We find a maximal collection of vertex-disjoint augmenting paths of length $2\ell+1$ in $G_Z$ that
are vertex-disjoint from paths already selected to be in $\Gamma$; we then add this new collection to $\Gamma$.
Since $G_Z$ is bipartite, this step can be done using the $\maximalaugpaths$ procedure from the
proof of Lemma~\ref{lem:mcm}.  Since we are working with $G_Z$ instead of $G$, when we find neighbors of a given vertex,
they are now restricted to be in $Z$ if the given vertex is in $[n]\setminus Z$, or vice versa; the data structure
$\DS_0^*$ allows for such queries.
The only other change is that when we initialize $\XX,\XX_1,\ldots,\XX_\ell$, we should remove vertices that have 
appeared in paths already selected to be in~$\Gamma$.  

Assume $2\ell+2 \le 1/\eps$.
We claim that after looping through all $Z\in \ZZ^{(n,1/\eps)}$, the resulting collection $\Gamma$ of vertex-disjoint
length-$(2\ell+1)$ augmenting paths is maximal in $G$.  To see this, consider any length-$(2\ell+1)$ augmenting path
$v_0u_1v_1\cdots u_\ell v_\ell u_{\ell+1}$  in $G$.  There exists $Z\in \ZZ^{(n,1/\eps)}$ such that
$u_1,\ldots,u_{\ell+1}\in Z$ and $v_0,\ldots,v_\ell\not\in Z$.
Thus, the path must intersect some path in $\Gamma$ during the iteration when we consider $Z$.
\end{proof}

We can now obtain a non-bipartite analog of Theorem~\ref{thm:main:mcm}:

\begin{theorem}\label{thm:main:mcm:gen}
Let $\CC$ be a class of geometric objects, where there is
a dynamic data structure $\DS_0$ for $n$ objects in $\CC$ that
can find an object 
intersecting a query object (if exists), and supports insertions and deletions of objects, with $O(\tau_0(n))$ query and update time.

Then there is a dynamic data structure for $O(n)$ objects in $\CC$ that maintains a $(1+O(\eps))$-approximation
of the maximum-cardinality matching in the intersection graph, under insertions and deletions of objects,
in $\OO(2^{O(1/\eps)} \tau_0(n))$ amortized time.
\end{theorem}
\begin{proof}
This is similar to the proof of Theorem~\ref{thm:main:mcm}, with Lemma~\ref{lem:mcm:gen} replacing
Lemma~\ref{lem:mcm}.  The only difference is that to support the data structure $\DS_0^*$, we maintain
$O(2^{O(1/\eps)}\log n)$ parallel instances of the data structure $\DS_0$ for $\XX\cap Z$ and
$\XX\setminus Z$, for every $Z\in \ZZ^{(n,1/\eps)}$.
This increases the update time by a factor of $O(2^{O(1/\eps)}\log n)$.

We have assumed that the input objects are labeled by integers in $[n]$.  When a new object is inserted, we can just
assign it the next available label in $[n]$.  When the number of objects exceeds $n$, we double $n$ and rebuild
the entire data structure from scratch.  Similarly, when the number of objects is below $n/4$, we halve $n$ and rebuild.
\end{proof}

\begin{corollary}\label{cor:mcm:disk2}
There is a dynamic data structure for $n$ disks in $\R^2$
that maintains a $(1+O(\eps))$-approximation of the maximum-cardinality matching
in the intersection graph, under insertions and deletions,
in $O(2^{O(1/\eps)}\log^{O(1)}n)$ amortized time.  
\end{corollary}

\begin{corollary}\label{cor:mcm:box2}
There is a dynamic data structure for $n$ axis-aligned boxes in $\R^d$ for any constant $d$
that maintains a $(1+O(\eps))$-approximation of the maximum-cardinality matching
in the intersection graph, under insertions and deletions,
in $O(2^{O(1/\eps)}\log^{O(1)}n)$ amortized time.  
\end{corollary}

\section{Conclusion}

In this paper, we have presented a plethora of results on efficient static and dynamic approximation algorithms for four fundamental geometric optimization problems, all obtained from two simple and general approaches (one for piercing and independent set, and the other for vertex cover and matching).  We hope that our techniques will find many more applications in future work on these and other fundamental geometric optimization problems.

Many interesting open questions remain in this area.  We list some below:
\begin{itemize}
\item Is there an $O(1)$-approximation polynomial-time algorithm for the piercing (MPS) problem for rectangles in $\R^2$ (like MIS for rectangles~\cite{mitchell2022approximating})?
Is there an $O(1)$-approximation polynomial-time algorithm for the weighted version of  independent set (MIS) for rectangles in $\R^2$?
Is there a sublogarithmic-approximation polynomial-time algorithm for (unweighted) MIS for boxes in $\R^3$?
If the answer is yes to any of these questions, our approach could automatically convert such an algorithm into near-linear-time static algorithms and efficient dynamic algorithms.
\item Is there an $O(n^\eps)$-approximation algorithm for MIS for arbitrary line segments or polygons~\cite{FoxP11} with near linear running time?  Our input rounding approach does not seem to work for arbitrary line segments or non-fat polygons.
\item Is there a $(2-\eps)$-approximation algorithm for vertex cover (MVC) for arbitrary line segments or strings~\cite{LokshtanovP00Z24} with near linear running time?
\item Are there efficient dynamic algorithms for the weighted version of geometric MVC similar to our unweighted MVC results?
\item Can we avoid the exponential dependence on $\eps$ in our results (in Section~\ref{sec:gen:match}) on non-bipartite geometric maximum-cardinality matching (MCM)? 
\end{itemize}







\old{
}



\newpage

\bibliographystyle{alphaurl}
\bibliography{main.bib}
%

\newpage

\appendix

\section{MPS and MIS: Speedup via Sampling}\label{app:samp}

In this appendix, we note that the $O(n^{1+O(\delta)})$ running time of our static algorithms for MPS and MIS in Sections~\ref{sec:mps}--\ref{sec:mis} can be lowered to $O(n\polylog n)$ if we are only required to output 
an approximation to the optimal value.
The idea is to approximate a sum by random sampling.
This idea has been used before in dynamic geometric algorithms, e.g., in \cite{ChanH21} on dynamic set cover.  In the weighted case, we use \emph{importance sampling} (which has also been used before in geometric optimization, e.g., in \cite{Indyk07}).

\begin{fact}\label{fact:samp}\ 

\emph{(Basic random sampling)}\ \
Let $a_1,\ldots,a_m\in [1,B]$.
Select a random multiset $R\subseteq \{a_1,\ldots,a_m\}$ of size $\ell = \lceil 4B/\eps^2\rceil$,
where each of its $\ell$ elements is chosen independently and is $a_i$ with probability $1/m$.
Then $X := \sum_{a_i\in R} a_i \cdot m/\ell$ is a $(1\pm \eps)$-approximation of $S := \sum_{i=1}^m a_i$ with probability at least $3/4$.

\emph{(Importance sampling)}\ \
More generally: Let $a_i\in [w_i,Bw_i]$ for each $i=1,\ldots,m$, with $\sum_{i=1}^m w_i=W$.
Select a random multiset $R\subseteq \{a_1,\ldots,a_m\}$ of size $\ell = \lceil 4B/\eps^2\rceil$,
where each of its $\ell$ elements is chosen independently and is $a_i$ with probability $w_i/W$.
Then $X := \sum_{a_i\in R} a_i\cdot W/(w_i\ell)$ is a $(1\pm \eps)$-approximation of $S:=\sum_{i=1}^m a_i$ with probability at least $3/4$.
\end{fact}
\begin{proof}
We will prove the more general statement.
First note that $\Ex[X] = \ell \sum_{i=1}^m (a_i \cdot W/(w_i\ell))\cdot w_i/W = S$.
Also, $\Var[X] \le \ell \sum_{i=1}^m (a_i \cdot W/(w_i\ell))^2\cdot w_i/W =
(1/\ell) \sum_{i=1}^m a_i^2W/w_i \le (1/\ell) \sum_{i=1}^m B a_i W  = BSW/\ell \le B S^2/\ell \le (\eps S/2)^2$.
By Chebyshev's inequality, $\Pr[|X-S|\ge \eps S]\le 1/4$.
\end{proof}

\begin{theorem}
Let $d$ be a constant.
Given a set $S$ of $n$ axis-aligned boxes in $\R^d$, 
we can compute an $O(\log\log \OPT)$-approximation to the size of the minimum piercing set for $S$ in $\OO(n)$ time w.h.p.\footnote{
With high probability, i.e., with probability $1-O(1/n^c)$ for an arbitrarily large constant $c$.
} by a static Monte Carlo randomized algorithm.
\end{theorem}
\begin{proof}
In the algorithm in the proof of Theorem~\ref{thm:pierce:box}, the size of the returned piercing set can be expressed as a sum of the sizes of piercing sets for a number of type-$d$ subproblems handled by Lemma~\ref{lem:pierce:box}.  The size of the piercing set for each type-$d$ subproblem lies in the range from 1 to $B:=O(b^{d})$ (after removing empty subproblems).  We can approximate the sum to within a constant factor by summing over a random sample of $O(B)$ terms
(see Fact~\ref{fact:samp}), with error probability at most $1/4$ (which can be lowered by
by repeating for logarithmically many trials and returning the median).
Thus, the number of calls to Lemma~\ref{lem:pierce:box} is $\OO(b^{d})$.
The total running time is now $\OO(n+b^{O(1)})$, which is $\OO(n)$ by setting $b=n^\delta$ for a sufficiently small constant $\delta$.
\end{proof}

\begin{theorem}
Let $d$ and $c$ be constants.
Given a set $S$ of $n$ objects in $\R^d$ of constant description complexity from a $c$-fat collection $\CCC$, 
we can compute an $O(1)$-approximation to the size of the minimum piercing set for $S$ in $\OO(n)$ time  w.h.p.\ by a static Monte Carlo randomized algorithm.
\end{theorem}
\begin{proof}
    This follows by a similar modification to the proof of Theorem~\ref{thm:pierce:fat} with random sampling.
\end{proof}

\begin{theorem}
Given a set $S$ of $n$ axis-aligned rectangles in $\R^2$, 
we can compute an $O(1)$-approximation to the size of the maximum independent set for $S$ in $\OO(n)$ time  w.h.p.\ by a static Monte Carlo randomized algorithm.

If the rectangles in $S$ are weighted, we can do the same for an $O(\log\log n)$-approximation to the weight of the maximum-weight independent set.
\end{theorem}
\begin{proof}
In the algorithm in the proof of Theorem~\ref{thm:indep:box} (after minor modification), the size/weight of returned independent set can be expressed as the maximum of the sizes/weights of $O((\log_b n)^d)$ independent sets, each of which is a sum of the sizes/weights of independent sets for a number of type-$d$ subproblems handled by Lemma~\ref{lem:indep:box}.  In the unweighted case, the size of the independent set for each type-$d$ subproblem lies in the range from 1 to $B := O(b^{d})$ (after removing empty subproblems).  In the weighted case, the weight of the independent set for the $i$-th type-$d$ subproblem lies between $w_i$ and $Bw_i$ where $w_i$ is the largest rectangle weight in the $i$-th subproblem.  We can approximate the sum to within a constant factor by summing over $O(B)$ terms, via basic random sampling in the unweighted case or importance sampling in the weighted case (see Fact~\ref{fact:samp}), with error probability at most $1/4$  (which can be lowered by
by repeating for logarithmically many trials and returning the median).  Thus, the number of calls to Lemma~\ref{lem:indep:box} is $\OO(b^{d})$. 
The total running time is now $\OO(n+b^{O(1)})$, which is $\OO(n)$ by setting $b=n^\delta$ for a sufficiently small constant $\delta$.
\end{proof}

\begin{theorem}
Let $d$ and $c$ be constants. 
Given a set $S$ of $n$ weighted objects in $\R^d$ of constant description complexity from a $c$-fat collection $\CCC$, 
we can compute an $O(1)$-approximation to the weight of the maximum-weight independent set for $S$ in $\OO(n)$ time  w.h.p.\ by a static Monte Carlo randomized algorithm.
\end{theorem}
\begin{proof}
    This follows by a similar modification to the proof of Theorem~\ref{thm:indep:fat} with importance sampling.
\end{proof}

\section{MVC: Fast Static Algorithms}

\subsection{Disks in $\R^2$}\label{app:disks}

In this subsection,
we give a static, near-linear-time, $(1+O(\eps))$-approximation algorithm $\AAA$ for MVC for disks under the promise that $n=O(\OPT)$.
As we are aiming for a $(1+O(\eps))$-factor approximation, 
we may tolerate an additive error of $O(\eps n)$
because of the promise.  MVC with $O(\eps n)$ additive error reduces
to MIS with $O(\eps n)$ additive error, by complementing the solution.

{\sf PTAS}s are known for MIS for disks, but allowing $O(\eps n)$ additive error, there are
actually {\sf EPTAS}s that run in near linear time.  For example, we can adapt an approach by Chan~\cite{Chan03}
based on divide-and-conquer via separators.
Specifically, we can use the following variant of Smith and Wormald's geometric separator theorem~\cite{SmithW98}:

\begin{lemma}[Smith and Wormald's separator]\label{lem:sep}
Given $n$ fat objects in a constant dimension $d$, there is an axis-aligned hypercube $B$, such that the
number of objects inside $B$ and the number of objects outside $B$ are both at most $(1-\beta)n$ for some constant $\beta>0$
(dependent on $d$),
and the objects intersecting $\partial B$ can be stabbed by $O(n^{1-1/d})$ points.  Furthermore, $B$ can be constructed in $O(n)$ time.
\end{lemma}
\begin{proof}
(The following is a modification of a proof in~\cite{ChanH24}, which is based on Smith and Wormald's original proof~\cite{SmithW98}.)

For each object, first pick an arbitrary ``center'' point inside the object.
Let $B_0$ be the smallest hypercube that contains at least $\frac{n}{2^d+1}$ 
center points.  Let $r$ be the side length of $B_0$.
Let $h$ be a parameter to be set later.
For each $t\in\{\frac 1h, \frac2h,\ldots, \frac{h-1}{h}\}$,
let $B_t$ be the hypercube with a scaled copy of $B_0$
with the same center and side length $(1+t)r$.
Since $B_t$ can be covered by $2^d$ quadtree boxes of side length $<r$,
the number of center points inside $B_t$ is at most $\frac{2^d n}{2^d+1}$.

An object of diameter $\le r/h$ intersects $\partial B_t$ for at most $O(1)$ choices of $t$.
Thus, we can find a value $t$ in $O(n)$ time such that $\partial B_t$ intersects $O(n/h)$ objects
of diameter $\le r/h$.  On the other hand, the objects of diameter $> r/h$ intersecting $\partial B_t$
can all be stabbed by $O(h^{d-1})$ points, because of fatness.  Set $h=n^{1/d}$, and $B_t$ satisfies the desired properties
with $\beta=\frac{2^d}{2^d+1}$.

One remaining issue is how to find $B_0$ quickly. 
Let $C$ be a sufficiently large constant.  Form a $C\times\cdots\times C$ grid, so that the number of center points between any two consecutive parallel grid hyperplanes is $n/C$.  This takes $O(C n)$ time by 
a linear-time selection algorithm.  Round each center point to its nearest grid point.  The rounded
point set is a multiset with $O(C^d)$ distinct elements.
Redefine $B_0$ as the smallest hypercube that contains at most $\frac{n}{2^d+1}$ rounded center points (multiplicity included),
which can now be computed in $C^{O(1)}$ time.  For any box, the number of center points inside the box
changes by at most $O(n/C)$ after rounding.  So, we just need to increase $\beta$ by $O(1/C)$.
\IGNORE{
For each object, first pick a ``reference'' point inside the object.
Compute the smallest quadtree box $B^*$ (a hypercube) that contains at least $2^d\beta n$ 
reference points; this can be found in $\OO(n)$ time by constructing and traversing the
(compressed) quadtree for the $n$ reference points.
Now, $B^*$ has $2^d$ child boxes (of half the diameter), and one of the child boxes, denoted by
$B_0$, contains at least $\beta n$ reference points.  Let $r$ be the diameter of $B_0$.
Let $b$ be a parameter to be set later.
For $t\in\{\frac 1b, \frac2b,\ldots, \frac{b-1}{b}$,
let $B_t$ be the hypercube with a scaled copy of $B_0$
with the same center and diameter $(1+t)r$.
Since $B_t$ can be covered by $3^d$ quadtree boxes congruent to $B_0$,
we see that the number of reference points inside $B_t$ is $6^d\beta n=(1-\beta)n$,
by choosing $\beta_d=1/(6^d+1)$.

An object of diameter $\le r/b$ intersects $\partial B_t$ for at most $O(1)$ choices of $t$.
Thus, we can find some $t$ in $O(n)$ time such that $\partial B_t$ intersects $O(n/b)$ objects
of diameter $\le r/b$.  On the other hand, the objects of diameter $> r/b$ intersecting $\partial B_t$
can be stabbed by $O(b^{d-1})$ points because of fatness.  Set $b=n^{1/d}$, and $B_t$ satisfies the desired properties.
}    
\end{proof}

The MIS algorithm is simple: we just recursively compute an independent set for the disks inside $B$
and an independent set for the disks outside $B$, and take the union of the two sets.
If $n$ is below a constant $b$, we solve the problem exactly by brute force in $2^{O(C)}$ time.
(This is analogous to Lipton and Tarjan's original {\sf EPTAS} for independent set for planar graphs~\cite{LiptonT77}.)

Since the optimal independent set can have at most $O(\sqrt{n})$
disks intersecting $\partial B$, the total additive error satisfies a recurrence of the form
\[ E(n) \le  \left\{\begin{array}{ll}
 \max_{n_1,n_2\le (1-\beta)n:\ n_1+n_2\le n}(E(n_1)+E(n_2)+O(\sqrt{n})) & \mbox{if $n\ge b$}\\
 0 & \mbox{if $n<b$.}
 \end{array}\right.\]
This yields $E(n)=O(n/\sqrt{b})$.  We set $b=1/\eps^2$.
The running time of the resulting static algorithm $\AAA$ is $T(n)=\OO(2^{O(1/\eps^2)}n)$.

\subsection{Rectangles in $\R^2$}\label{app:rect}

In this subsection, we give a static, near-linear-time, $(\frac32+O(\eps))$-approximation algorithm for MVC for rectangles, by adapting Bar-Yehuda, Hermelin, and Rawitz's previous, polynomial-time, $(\frac32+O(\eps))$-approximation algorithm~\cite{bar2011minimum}.

\paragraph{Triangle-free case.}
We first start with the case when the intersection graph is triangle-free, or equivalently, when
the maximum depth of the rectangles is at most~2.  (The depth of a point $q$ among a set of rectangles refers to 
the number of rectangles containing $q$; the maximum depth refers to the maximum over all $q\in\R^2$.)

We will design a static algorithm $\AAA$ for MVC for rectangles, under 
the promise that $n\le (2+O(\eps))\,\OPT$.  


We say that a rectangle $s$ \emph{dominates} another rectangle $s'$ if
$\partial s$ intersects $\partial s'$ four times, with $s$ having larger height than $s'$.
Bar-Yehuda, Hermelin, and Rawitz's algorithm proceeds as follows:

\begin{enumerate}
\item Decompose the input set $S$ into 2 subsets $S_1$ and $S_2$, such that there are no dominating pairs within each subset~$S_i$.

The existence of such a decomposition follows easily from Dilworth's theorem, but for an explicit construction,
 we can just define $S_1$ to contain all rectangles in $S$ that are not dominated
by any other rectangle in $S$, and define $S_2$ to be $S\setminus S_1$.
Correctness is easy to see (since the maximum depth is~2).

We can compute $S_1$ (and thus $S_2$) by performing $O(n)$ orthogonal range queries, after lifting each rectangle to 
a point in $\R^4$.  In fact, computing $S_1$ corresponds to the \emph{maxima} problem in $\R^4$, for which
$O(n\log n)$-time algorithms are known~\cite{ChanLP11}.

\item For each $i\in\{1,2\}$, compute a vertex cover $\XX_i$ of $S_i$ which approximates the minimum with additive
error at most $\eps n$.

Since we tolerate $\eps n$ additive error, it suffices to approximate the MIS of $S_i$
with $\eps n$ additive error.

Because $S_i$ has no dominating pairs, it forms a \emph{pseudo-disk} arrangement
(each pair's boundaries intersect at most twice).
It is known~\cite{bar2011minimum, KedemLPS86} that the intersection graph of any set of pseudo-disks
that have maximum depth 2 and have no containment pair is in fact planar.
(We can easily eliminate containment pair, by removing rectangles that contain another rectangle;
and we can detect such rectangles again by orthogonal range queries.)
So, we can use a known near-linear-time {\sf EPTAS} for MIS for planar graphs; for example, Lipton and Tarjan's
divide-and-conquer algorithm via planar-graph separators runs in $\OO(2^{O(1/\eps^2)}n)$ time~\cite{LiptonT77}.

\item Return $\XX$, the smaller of the two sets: $\XX_1\cup S_2$ and $\XX_2\cup S_1$.
\end{enumerate}

Let $\XX^*$ be the minimum vertex cover.
Then 
\begin{eqnarray*}
 |\XX| &\le& \tfrac{1}{2}(|\XX_1|+|S_2| + |\XX_2|+|S_1|)\\
 &\le& \tfrac{1}{2} (|\XX^*\cap S_1| + |\XX^*\cap S_2| + |S_1|+|S_2|) + \eps n
 \ =\ \tfrac{1}{2} (|\XX^*| + n) + \eps n\ \le\ (\tfrac32+O(\eps))|\XX^*|
\end{eqnarray*}
under the assumption that $n\le (2+O(\eps))|\XX^*|$.  The running time is
$T(n)=\OO(2^{O(1/\eps^2)}n)$.

At this point, if we apply Theorem~\ref{thm:main} with $c=\frac32+O(\eps)$, $\gamma=\eps$, and $\delta=\eps^3$ (and $\tau_0(n)=O(\log^{O(1)}n)$ as noted in Section~\ref{sec:specific}), we obtain:

\begin{corollary}\label{cor:rect:trifree}
There is a dynamic data structure for $n$ rectangles in $\R^2$
that maintains a $(\frac32+O(\eps))$-approximation of the minimum vertex cover
of the intersection graph, under insertions and deletions,
in $O(2^{O(1/\eps^2)}\log^{O(1)}n)$ amortized time, under the assumption that
the intersection graph is triangle-free.
\end{corollary}

\paragraph{General case.}
We now design our final static algorithm $\AAA$ for MVC for rectangles that avoids the triangle-free assumption, again by building on
Bar-Yehuda,  Hermelin, and Rawitz's approach:

\begin{enumerate}
\item Remove vertex-disjoint triangles $T_1,\ldots,T_\ell$ so that
the remaining set $S' := S\setminus (T_1\cup\cdots \cup T_\ell)$ is
triangle-free.

We can implement this step greedily in polynomial time, but a faster approach is via a plane sweep.
Namely, we modify the standard sweep-line algorithm to computing the maximum depth of $n$ rectangles in $\R^2$ (similar to Klee's measure problem)~\cite{preparata2012computational}.
The algorithm uses a data structure for a 1D problem: maintaining the maximum depth of intervals in $\R^1$,
subject to insertions and deletions of intervals.  Simple modification of standard search trees achieve
$O(\log n)$ time per insertion and deletion.  As we sweep a vertical line $\ell$ from left to right, we maintain
the maximum depth of the $y$-intervals of the rectangles intersected by $\ell$.  When $\ell$ hits
the left side of a rectangle, we insert an interval.  When $\ell$ hits the right side of a rectangle,
we delete an interval.  When the maximum depth becomes 3, we remove the 3 intervals containing the 
maximum-depth point, which correspond to a triangle in the intersection graph, and continue the sweep.
The total time is $O(n\log n)$.

\item Compute a vertex cover $\XX'$ of $S'$ which is a $(\frac32+O(\eps))$-approximation to the minimum.
This can be done by Corollary~\ref{cor:rect:trifree} in $\OO(2^{O(1/\eps^2)}n)$ time (we actually don't need
the full power of a dynamic data structure, since we just want a static algorithm for this step).

\item Return $\XX = \XX'\cup T$ with $T := T_1\cup\cdots\cup T_\ell$.
\end{enumerate}

Let $\XX^*$ be the minimum vertex cover.  
The key observation is that $\XX^*$ must contain at least 2 of the 3 vertices in each triangle $T_i$,
and so $|\XX^*\cap T|\ge \frac23 |T|$.
Thus,
\[ |\XX|\ =\ |\XX'|+|T|\ \le\ (\tfrac32+O(\eps))|\XX^*\cap S'| + \tfrac32|\XX^*\cap T|\ \le\ (\tfrac32+O(\eps))|\XX^*|.
\]
The running time of our final static algorithm $\AAA$ is $T(n)=\OO(2^{O(1/\eps^2)}n)$.

\end{document}